%% file: draft.tex
\documentclass[sigconf, nonacm]{acmart}

\newcommand\vldbdoi{XX.XX/XXX.XX}
\newcommand\vldbpages{XXX-XXX}
\newcommand\vldbvolume{14}
\newcommand\vldbissue{1}
\newcommand\vldbyear{2020}
\newcommand\vldbauthors{\authors}
\newcommand\vldbtitle{\shorttitle} 
\newcommand\vldbavailabilityurl{URL_TO_YOUR_ARTIFACTS}
\newcommand\vldbpagestyle{plain} 
\usepackage[hang,flushmargin]{footmisc} 
\hypersetup{draft}
\usepackage{float}
\usepackage{url}
\usepackage{multirow}
\usepackage{mathrsfs}
\usepackage{textcomp}
\usepackage{graphicx}
\usepackage{grffile}
\usepackage{mathtools, cuted}
\usepackage{dsfont}
\usepackage{amsthm}
\usepackage{diagbox}
\usepackage{subfig}
\usepackage{xcolor}
\usepackage{algorithmic}
\usepackage[ruled,linesnumbered,noend]{algorithm2e}
\usepackage{adjustbox}
\usepackage{diagbox}
\usepackage{bm}
\usepackage{xspace}
\usepackage{amsmath}
\usepackage{url}
\usepackage{hyperref}
\usepackage{amsfonts}
\usepackage{slashbox}
\newtheorem{example}{Example}
\newtheorem{theorem}{Theorem}
\newtheorem{lemma}{Lemma}
\newtheorem{definition}{Definition}
\usepackage{mathtools}
\usepackage{enumitem}
\usepackage[normalem]{ulem}

\SetCommentSty{mycommfont}

\newcommand{\nosemic}{\renewcommand{\@endalgocfline}{\relax}}
\newcommand{\dosemic}{\renewcommand{\@endalgocfline}{\algocf@endline}}
\let\oldnl\nl
\newcommand{\nonl}{\renewcommand{\nl}{\let\nl\oldnl}}

\usepackage{xcolor}

\definecolor{Mulberry}{rgb}{0.77,0.29,0.55}
\definecolor{CadmiumOrange}{rgb}{0.93,0.53, 0.18}
\definecolor{ForestGreen}{rgb}{0.13, 0.55, 0.13}
\definecolor{WildStrawberry}{rgb}{0.5, 0.7, 0.2}

\usepackage{amsmath}
\DeclareMathOperator*{\argmax}{arg\,max}

\makeatletter
\def\footnoterule{\kern-3\p@
  \hrule \@width 2in \kern 2.6\p@} 
\makeatother

\def\BibTeX{{\rm B\kern-.05em{\sc i\kern-.025em b}\kern-.08em
    T\kern-.1667em\lower.7ex\hbox{E}\kern-.125emX}}
\usepackage{extarrows}

\input{tex/macro.tex}

\begin{document}


\title{Influence Maximization in Real-World Closed Social Networks}

\author{Shixun Huang$^{\dagger \star}$, Wenqing Lin$^\ddagger$, Zhifeng Bao$^{\dagger *}$, Jiachen Sun$^\ddagger$}
\affiliation{
 \institution{\textsuperscript{$\dagger$}RMIT University, \textsuperscript{$\ddagger$}Tencent}
 \country{\textsuperscript{$\dagger$}\{shixun.huang,zhifeng.bao\}@rmit.edu.au, \textsuperscript{$\ddagger$}\{edwlin,jiachensun\}@tencent.com}
 }


 \begin{abstract}
In the last few years, many closed social networks such as WhatsAPP and WeChat have emerged to cater for people's growing demand of privacy and
independence. In a closed social network, the posted content is not available to all users or senders can set limits on who can see the posted content. 
 Under such a constraint, we study the problem of influence maximization in a closed social network. It aims to recommend users (not just the seed users) a limited number of \emph{existing} friends who will help propagate the information, such that the seed users' influence spread can be maximized. We first prove that this problem is NP-hard. Then, we propose a highly effective yet efficient method to augment the diffusion network, which initially consists of seed users only. The augmentation is done by iteratively and intelligently selecting and inserting a limited number of edges from the original network. Through extensive experiments on real-world social networks including deployment into a real-world application, we demonstrate the effectiveness and efficiency of our proposed method.

 \end{abstract}
  
\maketitle

\input{tex/introduction.tex}

\input{tex/relatedwork}

\input{tex/preliminary.tex}


\input{tex/methodforsingle.tex}

\input{tex/methodformultiple.tex}

\input{tex/shortenedexp.tex}

\section{Conclusion}
In this paper, we study the problem of Influence Maximization in Closed Social Networks which aims to recommend a limited number of edges for users to propagate information, such that  the seeds' influence via the selected edges is maximized. This problem is shown to be very useful in many industrial applications and we prove the NP-hardness of this problem. Moreover, we further propose a scalable and effective method to augment the diffusion network of seed users. We conduct extensive experiments to demonstrate that our method is very efficient and effective in our problem, a variant of our problem in open social networks and a real-world application. 
As a future direction, we will explore this problem under other diffusion models, and analyze theoretical properties and extension of our solutions.
 
 \smallskip
 \noindent
 \textbf{Acknowledgement}.
Zhifeng Bao is supported in part by ARC Discovery Project DP220101434
and DP200102611.

\newpage
\bibliographystyle{ACM-Reference-Format}
\bibliography{icde2019,library,IMlib,mypublications}

\newpage
\input{tex/appendix.tex}

\end{document}

%% file: tex/macro.tex
\newcommand{\short}{IMCSN}
\newcommand{\ksubnetwork}{$k$-subnetwork}
\newcommand{\method}[1]{\textsf{#1}}
\newcommand{\rppr}[1]{p^{RP}_{#1}}
\newcommand{\sppr}[1]{p^{SP}_{#1}}
\newcommand{\innermethod}{\textsf{PEI}}
\newcommand{\outtermethod}{\textsf{PSNA}}

%% file: tex/introduction.tex
\section{Introduction}\label{sec:intro}
Social network platforms have been a popular way that allows people to keep in touch with friends and share contents. There are many \emph{open} social networks where the posted content of a user will be available to all followers and even non-followers using search engines. This open sharing model is popular with millennials who are heavily impacted by the Fear Of Missing Out culture and like to put their lives on display and see everything that's going on~\cite{fomo}. 

Recently, \emph{closed} social networks, where sharing is limited to selected persons only, have emerged and are favored by generation Z to cater for privacy issues and information overload in open social networks. At Tencent, the closed sharing model has become a predominant form for sharing information in its numerous applications. For example, users will send a limited number of friends recent news and updates (e.g., COVID-19 statistics and reports of the Olympic winter games) on WeChat, charitable activities on QQ, and event invitations and sales promotion of merchandise on Tencent Shop. Meanwhile, many tech giants (e.g., Facebook, LinkedIn and Pinterest) have started enabling closed sharing models with personalized settings as well~\cite{trend,towardclosed1,towardclosed2,towardclosed3}. Additionally, the closed social network platform WhatsApp recently took the first place from Facebook in downloads~\cite{survery}, which again demonstrates the importance and popularity of the closed sharing model.\let\thefootnote\relax\footnotetext{$^\star$This work was done when the first author did an internship at Tencent.} \let\thefootnote\relax\footnotetext{$^*$Corresponding author.}

\begin{figure}[!t]

\centering
\includegraphics[width=0.2\textwidth]{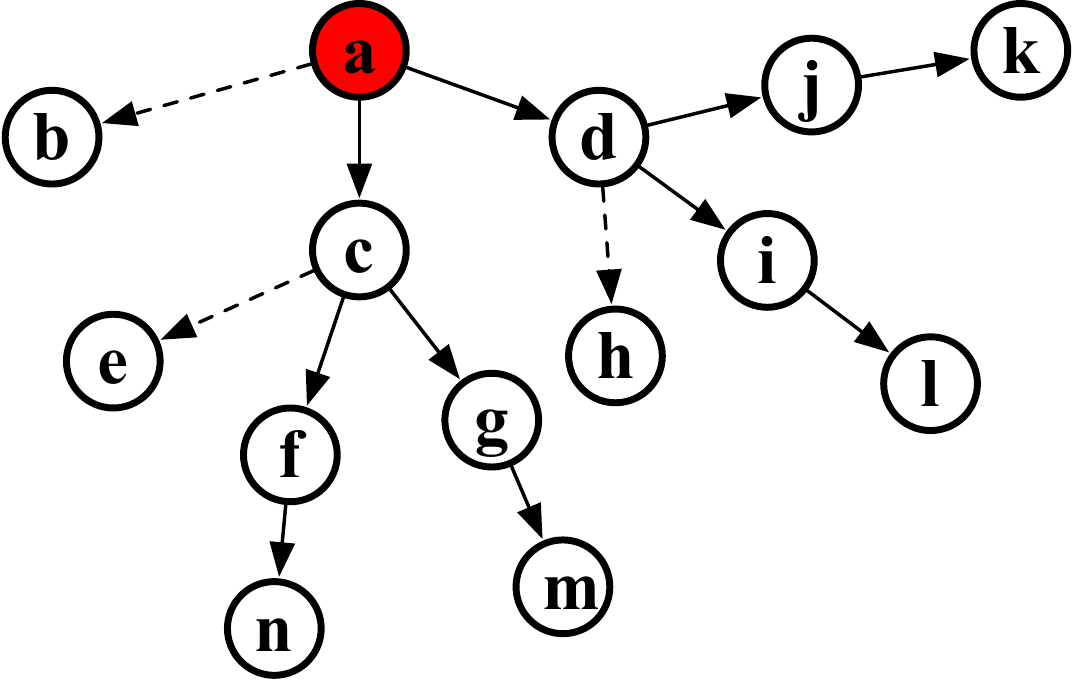}\hspace{-1em}
\caption{The network formed by solid edges is the optimal diffusion network to maximize the influence of $a$ if each user only shares information with up to two friends.}

\label{introexample}
\label{fig:intro}

\end{figure}

On the one hand, in closed social networks, users often do not want to overexpose themselves for various reasons, such as privacy concerns and psychological factors~\cite{expose}, or may have limited sharing opportunities due to resource constraints in events. Therefore, users are likely to share information with only a limited number of friends. 
For example, users in WeChat tend to share their life moments or private matters only with their close friends and families, rather than colleagues with pure business relationships~\cite{block,block2}.
Considering that online users have different capabilities of propagating information, as a consequence, the influence spread (i.e., information propagation effect) of a seed user (i.e., a user who releases the content at the very beginning) heavily depends on the friends she chooses to receive the information. Since the information propagation unfolds in a cascading manner, an information receiver can be multiple hops away from the seed and her selected friends for sharing also impacts the influence spread.

On the other hand, users hope that their released information could reach and benefit most of the users. Therefore, to fully realize the social and commercial value of the information with the greatest exposure in closed social networks, an effective strategy of recommending friends to share information is very important. Since each recommendation corresponds to a directed edge in the network, all recommendations constitute a \textit{diffusion network}, which is essentially a subnetwork of the original network and has limit on the number of outgoing edges (i.e., paths from senders to receivers) incident on each node. 

In this paper, we formulate the above recommendation as the problem of Influence Maximization in Closed Social Networks (\short{}), which aims to find the optimal diffusion network via which the seeds' influence spread is the maximum.

\begin{example}
To make this problem more intuitive, we use an example for illustration. Figure~\ref{introexample} shows a toy network where each user will be recommended a limited number of friends, say two, to share information. Suppose we only consider the information diffusion via the friends recommended by the system and ignore the uncontrollable and unpredictable sharing with friends that are not in the recommendation list. Given the seed user $a$ who has three friends $b$, $c$ and $d$, she will send the information to two of her friends, say $c$ and $d$. Afterwards, her friends who receive the information may also choose two of their friends (e.g., the friends $f$ and $g$ of $c$) to receive this information. This process is repeated such that the information diffuses in cascades. If we assume that users will propagate the received information unconditionally, the diffusion network formed by all solid edges refer to the optimal solution where the information from the seed user $a$ can reach the maximum number of users.
\end{example}

Notably, traditional influence maximization~\cite{Kempe03} in open social networks assumes that a user will send information to all her friends, and aims to select a limited number of seeds with the maximum influence spread. Here, we focus on the closed sharing model and selecting a limited number of edges such that the influence spread of \emph{specific seeds} via these \emph{selected edges} is maximized, which is different from the former. Please refer to Section~\ref{sec:relatedwork} for more details.

It is very challenging to solve the \short{} problem  because the choice space is extremely large.
To solve this problem, we try to insert limited links from the original network into the diffusion network that is initially empty.
Unfortunately, we find that greedily inserting edges with the maximum marginal gains w.r.t. influence of seeds is not effective, due to the non-submodular property of the objective function and the expensive marginal gain computations under the classical Independent Cascade (IC) model~\cite{Kempe03}. Thus, we resort to computing and leveraging influence lower bounds, and by using such lower bounds the submodular property can be preserved. Afterwards, we propose a novel diffusion network augmentation method which consists of two stages, namely the expansion stage and the filling stage. In the expansion stage, we iteratively expand the diffusion network size by incorporating the users and the connections that are important to spreading influence. In the filling stage, we intelligently fill up link recommendations for these involved users. Our contributions are summarized as below:


 \begin{itemize}[leftmargin=*]
\item We make the first attempt to study and formalize the problem of Influence Maximization in \emph{Closed} Social Networks (\short{}), motivated by numerous practical needs in social network platforms such as Tencent. We also prove its NP-hardness (Section~\ref{sec:preliminary}).

\item With an influence lower bound as the quality measurement, we first propose a network augmentation \emph{sketch} that can produce solutions for a single seed user with theoretical guarantees in a given diffusion network (Section~\ref{sec:methodsingle}.1). 
To boost the efficiency, we propose an effective yet scalable network augmentation method to avoid marginal gain computations (Section~\ref{sec:methodsingle}.2).  

\item We make some interesting observations in the influence diffusion of multiple seed users and leverage these observations to transform the \short{} problem for multiple seed users into the one for a single `virtual' seed user. Such a transformation enables us to utilize our method for a single seed user with minor adjustments, to produce effective solutions while maintaining high efficiency.

\item We conduct extensive experiments on real-world social networks including the deployment into a Tencent application (Section~\ref{sec:experiment}). We have several exciting findings: 

\noindent 1) In the \short{} problem, the boosted baselines, built upon the diffusion network produced by the expansion stage of our method, can achieve up to five-orders-of-magnitude larger influence spread than their counterparts built upon the initially empty network. Despite that, our full-stage method, which includes both the expansion and filling stages, significantly beats those boosted baselines. 

 2) Our full-stage method is able to identify important connections for spreading influence -- it is able to build a diffusion network where seed users can achieve 90\% of their full influence in the original network and this diffusion network contains only 36\% of edges from the original one.
 
 3) We deploy our solution into an activity of an online Tencent application where each online user will be recommended some friends for interaction.
 We conduct online A/B testing where each user is randomly assigned to one method which produces the recommendation list for this user. Such interaction can unfold in cascades and further trigger more interaction if the recommendation is effective. The online result shows that the recommendation from our solution which achieves a notably better Click-through Rate than the rest of baselines.

 
  4) In solving the problem of maximizing the seeds' influence over \emph{open} social networks, by recommending limited \emph{new links}, our method achieves up to five-orders-of-magnitude speedup than the state-of-the-art while maintaining competitive effectiveness.

\end{itemize}

%% file: tex/relatedwork.tex
\section{Related work}\label{sec:relatedwork}
In this section, we will first describe the difference between closed and open social networks. Afterwards, we will describe the related work on influence maximization in \emph{open} social networks, which can be broadly divided into two categories, influence maximization via node selection and influence maximization via edge insertion, respectively. Then, we will discuss the differences between them and our problem in \emph{closed} social networks.

\smallskip
\noindent \textbf{Open vs. Closed Social Networks}. The two words `open' and `closed' are used to describe the underlying sharing model rather than the topology of the social network. Furthermore, the sharing system applied to the network is decided based on the specific application behind. For example, if Figure~\ref{introexample} describes a subgraph of the Twitter social network, we have an open sharing model by default and a user's post will be available to all online users. On the other hand, if Figure~\ref{introexample} describes a subgraph of the Wechat social network, the closed sharing model will be applied in most cases (as in Example 1) where users will make their messages or posts visible to a limited number of selected friends. For ease of presentation and following the naming convention in this domain~\cite{naming1,naming2}, we directly use `open' and `closed' social networks to refer to the networks with the `open' and `closed' sharing model respectively.

\smallskip
\noindent \textbf{Influence Maximization via node selection}.
This problem refers to the classical influence maximization problem that aims to choose a limited number of seed nodes with the greatest influence spread. Kempe
et al.~\cite{Kempe03} prove the NP-hardness of this problem and propose a Monte Carlo simulation based greedy algorithm which iteratively chooses the node with the greatest marginal gain to the influence. Due to the importance of this problem, many subsequent studies
~\cite{huang2019finding,huang2021capturing,leskovec2007cost,goyal2011celf++,cheng2013staticgreedy,ohsaka2014fast,borgs2014maximizing,nemhauser1978analysis,tang2014influence,goyal2011simpath,jung2012irie,galhotra2016holistic,chen2010scalable,chen2009efficient,chen2010scalable2,wang2017real,bian2020efficient,huang2020efficient,bevilacqua2021fractional,cautis2019adaptive,zhang2020geodemographic,aslay2018influence,cai2020target,huang2019community,liu2019influence,li2019tiptop,he2019tifim,sun2018multi,wu2018two,minutoli2019fast,banerjee2019combim,han2018efficient,tang2015influence} 
have been proposed to further improve the efficiency and/or effectiveness.  They mainly differ in how the influence spread is defined or estimated. To name a few, \citeauthor{leskovec2007cost}~\cite{leskovec2007cost} adopt the Monte-Carlo simluation to measure the influence spread and speed up the process in~\cite{Kempe03} with an early termination technique based on the submodular property of the influence function. \citeauthor{ohsaka2014fast}~\cite{ohsaka2014fast} measure the influence and marginal gain based on a limited number of subgraphs generated by the flipping-coin technique~\cite{Kempe03}, which further improves the efficiency. \citeauthor{borgs2014maximizing}~\cite{borgs2014maximizing} leverage reverse reachable sets to estimate influence and inspires more recent advanced solutions~\cite{tang2014influence,tang2015influence}.

\smallskip
\noindent \textbf{Influence maximization via edge insertion}.
The work falling in this category aims to insert a limited number of edges into the network so as to maximize influence spread of specific seeds. \citeauthor{d2019recommending}~\cite{d2019recommending} study how to maximize the influence of the seeds under the IC model by adding a limited number of edges incident to these seeds. \citeauthor{coro2021link}~\cite{coro2021link} extend the results of ~\cite{d2019recommending} to the Linear Threshold (LT) model~\cite{granovetter1978threshold}. Specifically, they prove that the objective function under the LT model is submodular and leverage this property to propose an approximate algorithm. 
\citeauthor{khalil2014scalable}~\cite{khalil2014scalable} study how to add a limited number of edges to maximize the influence of given seeds under the LT model. 
\citeauthor{chaoji2012recommendations}~\cite{chaoji2012recommendations} study how to add edges to maximize the influence spread of seeds under the constraint that at most $k$ inserted edges are incident on any node.
\citeauthor{yang2019marginal}~\cite{yang2019marginal} and \citeauthor{yang2020boosting}~\cite{yang2020boosting} study how to add a limited number of edges from a candidate set to maximize the seeds' influence under the IC model. \citeauthor{yang2019marginal}~\cite{yang2019marginal} derive a lower and upper bound influence function respectively to approximate the non-submodular influence under the IC model, and they use a sandwich strategy to produce approximate solutions. \citeauthor{yang2020boosting}~\cite{yang2020boosting} derive tighter bounds than \citeauthor{yang2019marginal}~\cite{yang2019marginal} and produce more effective results.

\smallskip 

\emph{Differences}.~The aforementioned studies are drastically different from ours. The main reason is on the problem setting and assumption. 
The studies of the first category focus on selecting \emph{seeds} based on the open social networks where the influence from the seeds is allowed to propagate via \emph{any} edge in the network. On the other hand, we consider the \emph{closed} social networks, and focus on selecting limited \emph{edges} from the original network such that the influence of \emph{specific seeds} via these \emph{selected edges} is maximized.

The studies of the second category focus on inserting a limited number of \emph{new} edges into the \emph{existing} network and assume an \emph{open-sharing} model such that the influence of the seeds will spread via all edges including the inserted ones. 
In contrast, we focus on inserting a limited number of \emph{existing} edges from the original network into an initially empty network and consider a \emph{closed-sharing} model such that the influence of seeds will spread \emph{only} via inserted edges. As a consequence, existing work's decision making of edge insertion is based on the current graph structure and cannot be trivially extended to handle our problem, because the network which requires edge insertion in our problem is initially empty. 
Additionally, due to different assumptions of the sharing model, the space of candidate edges for insertion that is considered by existing work is significantly smaller than that of our problem. As a result, most existing methods will suffer from serious scalability issue even if they could be extended to handle our case. In particular, the infeasibility of extending existing work can be caused by reasons including but not limited to: (1) Most of these studies do not consider the edge insertion constraint where the number of inserted edges sharing the same source node cannot be greater than $k$. (2) The methods in~\cite{d2019recommending,coro2021link} are designed based on the assumption that the source nodes of all edge candidates must be seeds, which is not the case in our problem. (3) The method in~\cite{khalil2014scalable} is specifically designed for the LT model which has drastically different properties from the IC model we are considering. (4) The method in~\cite{chaoji2012recommendations} cannot help the seed nodes to influence the nodes that are 2-hop away in the initially empty network, because the influence path between any pair of nodes is not allowed to contain more than one inserted edge. (5) The methods in~\cite{yang2019marginal,yang2020boosting} hold a strong assumption that the network is acyclic, which is often not true in real-world scenarios.

%% file: tex/preliminary.tex
\section{Problem Formulation}\label{sec:preliminary}
In this paper, we consider both directed and undirected social network where the latter can be transformed into the directed one. In an undirected network, each edge represents a friendship between two users. Since information diffusion is directed, each friendship corresponds to two diffusion directions. Thus, we represent the undirected social network as the directed one $G=(V,E)$ where $V$ ($E$) is the node (edge) set, and each directed edge from $u$ to $v$ is denoted as $(u,v)$. We use $N^G_{in}(u)$ and $N^G_{out}(u)$ to denote the set of incoming and outgoing neighbors of $u$, respectively.

\smallskip
\noindent \textbf{The diffusion model}. We focus on a classic and widely-adopted
information diffusion model -- the {\textit{Independent Cascade}} (IC)
model~\cite{Kempe03}.
It originates from the marketing
literature~\cite{goldenberg2001using} and independently assigns each
edge $(u,v)$ with an influence probability $p_{u,v}
\rightarrow [0,1]$. Given a seed node $s$ being active at time step 0 and the influence probabilities, the diffusion unfolds in discrete steps.
Each active node $u$ in time 
step $t \geq 1$ will have a single chance to activate each
outgoing neighbor $v$, that is inactive in step $t-1$, with a probability of $p_{u,v}$.
If an outgoing neighbor $v$ is activated in step $t$, it
will become active in step $t+1$ and then will have a single chance
to activate each of its inactive outgoing neighbors in the next time step.
The diffusion instance terminates when no more nodes can
be activated.
The \emph{influence spread} $\delta_G(s)$ in a graph $G$ is the expected number of activated nodes with $s$ as the seed node. Note that our proposed methods are specifically designed for the IC model. The extension on other models (e.g., the Linear Threshold model~\cite{Kempe03}) is out of the scope of this work and will be explored in future work.

\smallskip
In this paper, we aim to select at most $k$ outgoing edges for each online user such that the influence spread of seeds via these selected edges is maximized. These directed edges naturally form a diffusion sub-network defined as below. 
  
\begin{definition}[K-subnetwork]
Given a directed network $G=(V,E)$ and an integer $k$, a k-subnetwork $G_k$ is a subgraph of $G$ where there are at most $k$ outgoing neighbors for each $v\in G_k$. Formally, $G_k=(V_k,E_k)$ where $V_k\subseteq V$, $E_k \subseteq E$ and $\forall v \in V_k, |N^{G_k}_{out}(v)|\leq k$.

\end{definition}  
 
For example, in Figure~\ref{introexample}, the network that consists of solid edges only is a 2-subnetwork.
%
Notably, based on industry practice, we make three considerations when formulating our problem (i.e., Definition~\ref{defn:problem}): 1) The value of $k$ is the same for all users and is decided by the event operator. However, our solutions can easily work with the scenario where the value of $k$ depends on specific users. 
 2) Seed users are independent since KOLs (i.e., key opinion leaders) may have unpredictable and different information to spread over a long period, and recommendations based on this consideration help maintain a long-term usage of the recommendation system. 
 3) The friend recommendation list for a user is built by considering all seed users instead of being customized for individuals since (i) it is intractable to know, among all seed users sharing the same content, who will activate online users for further information propagation in real-world scenarios, and (ii) fixing the recommendations help increase the interaction rate between users with a laser focus on facilitating information diffusion.

\begin{definition}[Influence maximization in closed social networks (\short{})]\label{defn:problem}
Given a directed social network $G=(V,E)$, an integer $k$, a set $S$ of independent seed users, we aim to find the optimal diffusion \ksubnetwork{} $G^*_k$ such that the aggregated influence of seed users in $S$ is maximized under $G^*_k$:

$$  G^*_k = \argmax_{G_k \in \mathbb{S}}  \sum_{s \in S} \delta_{G_k} (s) $$ 

\noindent where $\mathbb{S}$ refers to the whole space of all possible k-subnetworks of $G$.
\end{definition}

If $s \notin V_k$,  the influence of $s$ in $G_k$ is 0. When the context is clear, we omit the subscript of $\delta_{G_k}(\cdot)$, and we use the terms `node' and `user', as well as the terms `edge' and `link' interchangeably. 

\begin{figure}[!t]

\centering
\includegraphics[width=0.3\textwidth]{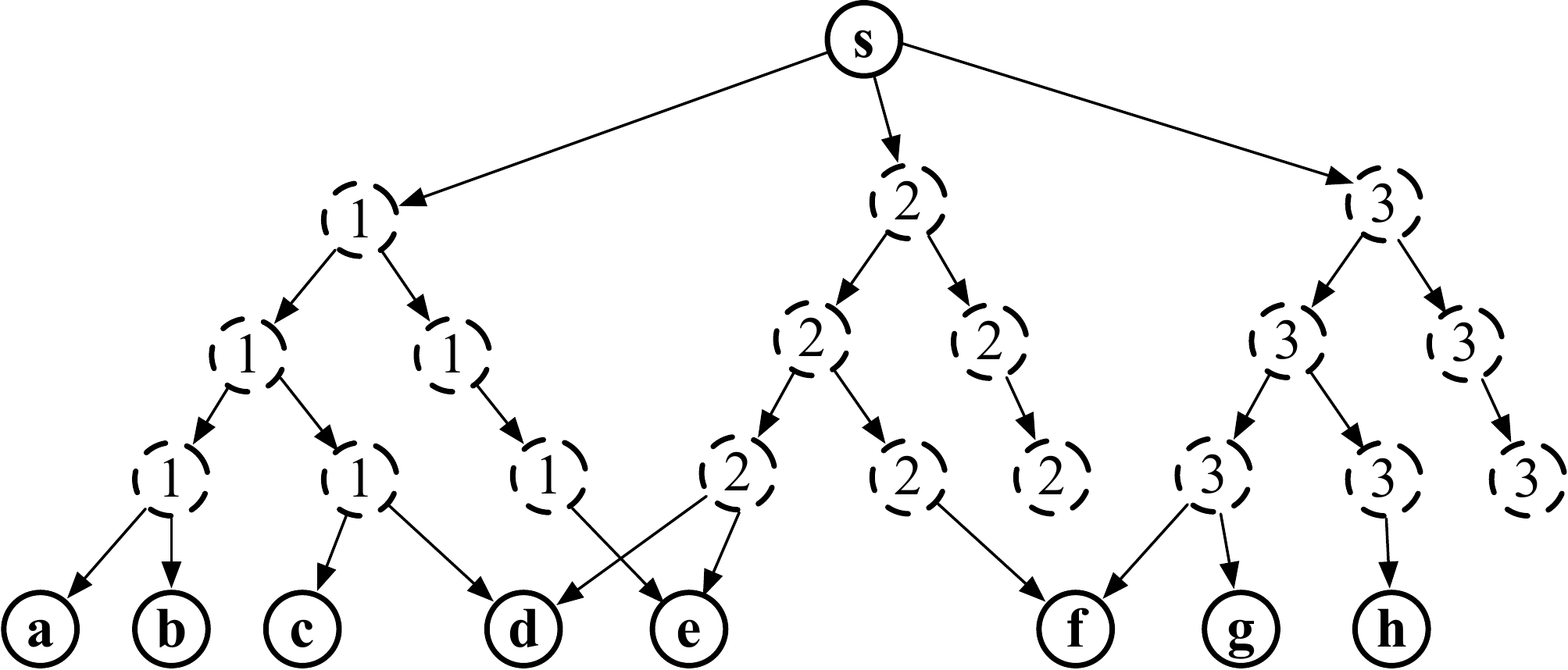}\hspace{-1em}
\caption{An example of a graph constructed based on a set cover instance.} 
\label{fig:tree}
\end{figure}

\smallskip
 \noindent \textbf{Hardness analysis of the \short{}}. We will show that the \short{} problem is NP-hard via a reduction from the set cover problem defined as below.

\begin{definition}[Set Cover]
Given a set $U$ of elements, a collection $C$ of subsets $A_1,A_2,...,A_{|C|}$ of $U$ where $\cup_{1\leq i\leq |C|} A_i = U$, we aim to choose the smallest number of sets from $C$ such that their union is $U$.
\end{definition}

\begin{theorem}\label{proof:nphard}
The \short{} problem is NP-hard.
\end{theorem}

\begin{proof}
To show that \short{} is NP-hard, we will perform a reduction from the NP-complete decision problem of set cover, which checks whether we can find $k$ sets from $C$ whose union is $U$, to the special case of \short{} where there is only one seed user $s$.

Given the collection $C$ of a set cover instance, we construct a deterministic graph with a tree structure where the weight of each edge is 1 as below. First, we find the set $A_{max}$ with the greatest size from $C$ and create a tree with elements in $A_{max}$ as the bottom nodes (at the bottom level). In the second to last level, we introduce the set $V'$ of $\lceil |A_{max}|/k  \rceil$ virtual nodes. We next create directed edges from each of the virtual nodes to at most $k$ arbitrary nodes in $A_{max}$ which have not been connected by virtual nodes in $V'$. Afterwards, we create the third to last level based on the second to last level by adopting the similar rules. We repeat this step until the root node is constructed. We also create trees for the rest of the sets in $C$ similarly. One difference is that, when we create the tree for a set $A$, we need to make sure that the number of virtual nodes in each level is consistent with the corresponding level of the tree constructed based on $A_{max}$. It means that some virtual nodes may not have outgoing neighbors. After we create trees for all sets in $C$, we introduce directed edges from $s$ to all root nodes.

Figure~\ref{fig:tree} shows an example where a graph is constructed from the set cover instance where $k=2$ and $C=\{ \{a,b,c,d,e\},\{d,e,f\}, \allowbreak \{f,g,h\} \}$. Dashed nodes refer to virtual nodes. Virtual nodes with the same number come from the tree constructed based on the same set from $C$. Since $A_{max}=\{a,b,c,d,e\}$, the number of virtual nodes in the second to last level of all trees are $\lceil |A_{max}|/k  \rceil =3$. 


Since the time cost for creating such a tree is at most $O(|A_{max}| \allowbreak \log_k|A_{max}|)$, the total reduction process is polynomial in the total size of the sets in the collection. To solve the \short{} problem in the constructed graph, we only need to focus on selecting $k$ outgoing edges for the root node $s$ since the out-degree of the rest of the nodes in the graph is at most $k$ and edges not incident on $s$ can all be chosen. Considering that the trees constructed from each set in $C$ have the same number of virtual nodes, the quality of the outgoing edge selection for $s$ only depends on bottom nodes being reached. Thus, selecting $k$ outgoing edges for $s$ corresponds to selecting $k$ sets from $C$. Suppose the number of virtual nodes in each tree is $x$. If the optimal influence spread is $k\times x +1+|U|$ (including $s$) which is also the maximum possible influence, we have a `Yes' answer to the set cover instance. Otherwise, the answer is `No'. Therefore, if we can find the optimal solution for the \short{} instance in polynomial time, the set cover problem can be solved in polynomial time which is not possible unless P=NP. Thus, the \short{} problem is NP-hard.
\end{proof}



%% file: tex/methodforsingle.tex
\section{Solving \short{} for a single seed}\label{sec:methodsingle}
To facilitate the illustration of our methodology, we first describe how to solve the \short{} problem for 
a single seed user. In the next section, we will describe how to extend the method to handle multiple seed users. 
Essentially, we try to solve the \short{} problem by inserting edges from the original network into a \ksubnetwork{} which initially contains the seed user only. As mentioned in Section~\ref{sec:intro}, we resort to the influence lower bounds to preserve submodularity, and the lower bounds are computed with the Restricted Maximum Probability Path (RMPP) model~\cite{chaoji2012recommendations} as defined below. 


\begin{definition}[Influence Probability of a path]
The influence probability of a path is the product of influence probabilities of all edges in this path.
\end{definition}

\begin{definition}[Restricted Maximum Probability Path (RMPP)~\cite{chaoji2012recommendations}]\label{defn:rmpp}
Given a \ksubnetwork{} $G_k$ where edges are either native (i.e., originally exist) or inserted, and a seed node $s$, the restricted maximum probability path $RMPP(s,u)$ is the directed path from $s$ to $u$ whose probability
$\rppr{(s,u)}$ is the maximum among all paths from $s$ to $u$ containing at most one inserted edge. Ties are broken arbitrarily.
\end{definition}

\begin{definition}[RMPP model~\cite{chaoji2012recommendations}]
In the RMPP model, the influence of a seed node $s$ to a node $u$ is the influence probability $\rppr{(s,u)}$ of $RMPP(s,u)$, and the total influence of $s$ is the sum of influence probabilities of all RMPPs from $s$ to the rest of nodes in the graph.
\end{definition}

\begin{figure}[!t]

\centering
\includegraphics[width=0.2\textwidth]{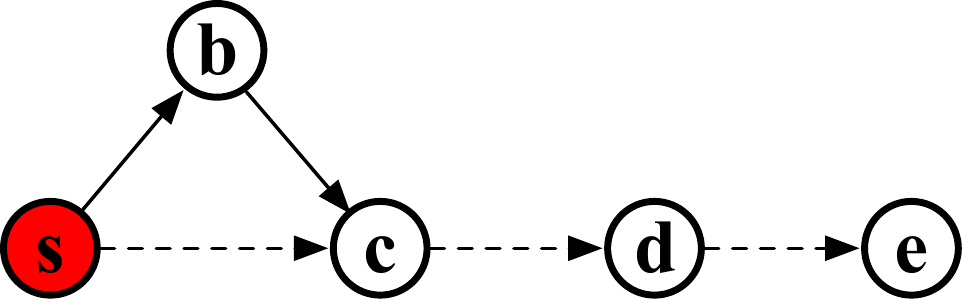}\hspace{-1em}
\caption{A simple network where dashed edges are candidates for insertion.} 
\label{fig:rmppexample}
\end{figure}

\begin{example}\label{rmpp_example}
As shown in Figure~\ref{fig:rmppexample}, we have a network which contains five nodes and two native edges $(s,b)$ and $(b,c)$ with the same influence probability of 0.5. Suppose the seed node is $s$, three edge candidates $(s,c)$,$(c,d)$ and $(d,e)$ have the same influence probability of 1,  $T_1=\emptyset$ and $T_2=\{(s,c),(c,d)\}$. 
The marginal gain of $(s,c)$ over $T_1$ is 1-0.25=0.75 since $RMPP(s,c)$ is edge $(s,c)$. However, the marginal gain of $(d,e)$ over $T_2$ is 0 since $s$ needs to traverse at least two inserted edges to reach $e$, which is not allowed in the RMPP model. 
\end{example}

\begin{algorithm}[!t]
\small
 \caption{\method{SubnetworkAugmentationSketch}}\label{alg:initialframework}
   \SetKwInOut{Input}{Input}
   \SetKwInOut{Output}{Output}
 \Input{The input network $G$, the seed user $s$ and an integer $k$.}
  \Output{The \ksubnetwork{} $G_k$.}
  
$G_k=(V_k.E_k) \leftarrow $ \ksubnetwork{} containing the seed $s$ only\;
 
 \While{ $\exists v \in V_k , |N_{out}^{G_k}(v) | < min(k,|N_{out}^{G}(v)| )$   }{
  $G_k=$ {\method{EdgeInsertionSketch}}$(G,G_k,s,k)$\;
 } 
 
  Return $G_k$\; 
\end{algorithm}

RMPP provides an efficient lower bound estimation of influence spread under the IC model and preserves the submodularity~\cite{chaoji2012recommendations}. 
Unfortunately, directly adopting the RMPP model to insert edges into the \ksubnetwork{} which initially contains only the seed node $s$ does not help to achieve great influence spread. This is because $s$ has to traverse two inserted edges to reach the nodes that are two hops away, but this is not permitted in the RMPP model, as illustrated in Example~\ref{rmpp_example}. 

To mitigate this issue without missing the opportunity of leveraging the submodularity property of the RMPP model, we propose a subnetwork augmentation sketch in Algorithm~\ref{alg:initialframework}. Specifically, this sketch iteratively calls a strategy called \method{EdgeInsertionSketch} in the pseudocode. This strategy tries to insert edges to increase the influence spread of the seed node $s$ under the RMPP model and uses these inserted edges to update the input \ksubnetwork{} $G_k$. After reaching a termination condition, these \emph{inserted} edges will be treated as \emph{native} edges of $G_k$ in the next iteration, which will help $s$ to influence the nodes that are many hops away. In what follows, we will introduce how the method \method{EdgeInsertionSketch} works and then propose a much more efficient yet practical edge insertion strategy (Section~\ref{sec:edgeinsertion}). Afterwards, we will further optimize this subnetwork augmentation sketch (Section~\ref{sec:augmentation}).

\subsection{Edge Insertion}\label{sec:edgeinsertion}
Given a candidate network $G_c$ that contains candidate edges to be inserted into the current \ksubnetwork{} $G_k$, the most straightforward way to maximize the influence of $s$ is to greedily insert an edge with the maximum marginal gain into $G_k$ without breaking the constraint of \ksubnetwork{}. This approach suffers from serious scalability issues since it needs to update the marginal gain of all candidate edges in each iteration and the marginal gain computation of each edge incurs $O((V_k+E_k)\log V_k)$ time complexity with a variant of Dijkstra algorithm as in~\cite{chaoji2012recommendations}. To address this issue, we first present an edge insertion sketch from the theoretical point of view, followed by a practical insertion method.


 \begin{algorithm}[!t]
 \small
  \caption{ \method{EdgeInsertionSketch}}\label{alg:networkaugmentation}
    \SetKwInOut{Input}{Input}
    \SetKwInOut{Output}{Output}
   \Input{The input network $G=(V,E)$, the current \ksubnetwork{} $G_k=(V_k,E_k)$, the seed user $s$ and an integer $k$.}
   \Output{The updated \ksubnetwork{} $G_k$.}
 
 $\mathcal{C} \leftarrow $candidate set initialized as $\emptyset$\;
 \ForEach{ node $u \in G_k$}{
   \If{$|N^{G_k}_{out}(u)|<min(k,|N^{G}_{out}(u)|)$}{
     \ForEach{ $(u,v) \in E$ }{
       \If{$(u,v) \notin E_k $ and $v \neq s$}{
         Add $(u,v)$ into $\mathcal{C}$
       }
     
     }
   }
 
 }
 
 $G_c=\{V_c,E_c\} \leftarrow $ construct candidate graph based on $\mathcal{C}$\;
 
 Compute the probabilities of SMPPs from $s$ to all nodes in $G_k$ \;
 
 \ForEach{node $v \in V_c$ with $|N^{G_c}_{in}(v)|>0$}{
   $u^*_v= \argmax_{u \in N^{G_c}_{in}(v)} \sppr{(s,u)} \cdot p_{(u,v)} $\;
 }
 
 \While{$G_c$ has nodes with incoming neighbors}{
        
        $v^*= \argmax_{v \in V_c, |N^{G_c}_{in}(v)|>0}  \delta_{G_k \cup \{(u^*_v,v)\} }^\triangle (s) -  \delta_{G_k}^\triangle (s)$\;

     Remove all incoming edges to $v^*$ from $G_c$\; 
     Insert $(u^*_{v^*},v^*)$ into $G_k$\;
     
     \If{$|N^{G_k}_{out}(u^*_{v^*})|=min(k,|N^{G}_{out}(u^*_{v^*})|)$}{
       \ForEach{$v \in N^{G_c}_{out}(u^*_{v^*})$}{
         \lIf{$u^*_v=u^*_{v^*}$}{
                 Update $u^*_v$
         }
         Remove $(u^*_{v^*},v)$ from $G_c$\;
       }
     }
 }
 
 Mark all inserted edges as native\;

 \Return $G_k$\;

 \end{algorithm}

 \subsubsection{Edge Insertion Sketch}
 
Later in Theorem~\ref{theorem:critical}, we prove that for all candidate edges that share the same destination node, we can carefully insert the \emph{critical} edge (which will be defined in Definition~\ref{defn:criticalneighborandedge}) such that inserting other edges after this critical edge will not increase the influence spread of $s$ under the RMPP model. With this observation, we only need to compare the marginal gain of critical edges with different destination nodes in each iteration. Hence, the cost of marginal gain computation can be greatly reduced. 
To facilitate the presentation of our method, we use $\delta^\triangle()$ as the influence under the RMPP model, and introduce a concept, namely \emph{Strict Maximum Probability Path (SMPP)}. 

\begin{definition}[Strict Maximum Probability Path (SMPP)]\label{def:smpp}
  Given a seed node $s$ and a \ksubnetwork{} $G_k$ where edges are either native or inserted, the strict maximum probability path $SMPP(s,u)$ is a directed path from $s$ to $u$ whose probability $\sppr{(s,u)}$ is the maximum among all the paths from $s$ to $u$ containing native edges only (i.e., a special case of RMPP). Ties are broken arbitrarily.
  \end{definition}
  
 \begin{example}
As shown in Figure~\ref{fig:rmppexample}, $SMPP(s,c)$ consists of two native edges, namely $(s,b)$ and $(b,c)$. $SMPP(s,c)$ remains unchanged even after we insert edge $(s,c)$, since the inserted edge is not native.
 \end{example}

The algorithm to be introduced is based on Theorem~\ref{theorem:critical} which is in turn established upon  Lemma~\ref{theorem:pair} and Lemma~\ref{theorem:marginalgain} below. 

\begin{lemma}\label{theorem:pair}
Given a \ksubnetwork{} $G_k$, a seed node $s$, and two candidate edges $(u,v)$ and $(u',v)$ which are in the candidate network $G_c$ and share the same endpoint $v$, and $\rppr{(s,x)}(u,v)$ which denotes the probability of $RMPP(s,x)$ after inserting the edge $(u,v)$ into $G_k$, if $ \sppr{(s,u)}\cdot p_{(u,v)} \geq \allowbreak \sppr{(s,u')} \cdot p_{(u',v)} ,$
then for any node $x$ which can be reached by $v$ via directed paths, 
$\rppr{(s,x)}(u,v)  \geq \rppr{(s,x)}(u',v) $.
\end{lemma}

\begin{proof}
If the $RMPP(s,x)$ from $s$ to $x$ remains the same after inserting $(u,v)$ and $(u',v)$, the influence probability of $x$ will remain unchanged, i.e., $\rppr{(s,x)}(u,v) = \rppr{(s,x)}(u',v)$.

If the $RMPP(s,x)$ from $s$ to $x$ is changed after inserting $(u,v)$ or $(u',v)$, the $RMPP$ must contain $(u,v)$ or $(u',v)$. Let $RMPP_{(u,v)}$ and $RMPP_{(u',v)}$ denote the \emph{updated} $RMPP$ from $s$ to $x$ by inserting edge $(u,v)$ and $(u',v)$ respectively. Each of these two $RMPP$s consist of two sub-paths. That is, $P_1(RMPP_{(u,v)})=[s,...,u,v]$,  $P_2(RMPP_{(u,v)})=[v,...,x]$, $P_1(RMPP_{(u',v)})=[s,...,u',v]$ and $P_2(RMPP_{(u',v)})=[v,...,x]$. Based on Definition~\ref{defn:rmpp}, the second paths of both 
$RMPP$s do not contain inserted edges. Thus, $P_2(RMPP_{(u,v)})$ and $P_2(RMPP_{(u',v)})$ have the same probability. 
Since $\sppr{(s,u)}\cdot p_{(u,v)}$ and $\sppr{(s,u')}\cdot p_{(u',v)} $ refer to the probability of the first sub-paths of these two RMPPs respectively and the former one is not smaller than the latter, we have $\rppr{(s,x)}(u,v)  \geq \rppr{(s,x)}(u',v)$.
\end{proof}

\begin{lemma}\label{theorem:marginalgain}
Given a \ksubnetwork{} $G_k$, a seed node $s$, and two candidate edges $(u,v)$ and $(u',v)$ which are in the candidate network $G_c$ and share the same endpoint $v$, if $\sppr{(s,u)}p_{(u,v)} \geq \allowbreak \sppr{(s,u')} p_{(u',v)} $, $\delta_{G_k \cup \{(u,v)\}}^\triangle(s)=\delta_{G_k \cup \{(u,v),(u',v) \}}^\triangle(s)$.
\end{lemma}

\begin{proof}
Based on Lemma~\ref{theorem:pair}, for any node $x$ that can be reached by $v$, the influence probability $\rppr{(s,x)}$ of $RMPP(s,x)$ under $G _k\cup \{(u,v)\}$ is never smaller than that under $G_k \cup \{(u',v)\}$. 
Thus, introducing $\{(u',v)\}$ after inserting $(u,v)$ does not contribute to the influence increment of $s$. Thus, $\delta_{G_k \cup \{(u,v)\}}^\triangle(s)=\delta_{G_k \cup \{(u,v),(u',v) \}}^\triangle(s)$.
\end{proof}

To better describe Theorem~\ref{theorem:critical} later, we have the definition below.

\begin{definition}[critical neighbor and edge]\label{defn:criticalneighborandedge}
Given a \ksubnetwork{} $G_k$, the seed node $s$, the candidate network $G_c$ and a node $v$, if $|N^{G_c}_{in}(v)|>0$, the critical neighbor $u^*_v$ of $v$ in $G_c$ is the neighbor whose SMPP path probability $\sppr{(s,u^*_v)}$ times the influence probability $p_{(u^*_v,v)}$ of the edge $(u^*_v,v)$ is the greatest among all the incoming neighbors of $v$ in $G_c$. That is, $u^*_v= \argmax_{ u \in N^{G_c}_{in}(v)} \sppr{(s,u)}\cdot p_{(u,v)}$. Correspondingly, the incoming edge $(u^*_v,v)$ is the critical edge of $v$.

\end{definition}

\begin{theorem}\label{theorem:critical}
Given a \ksubnetwork{} $G_k$, the seed node $s$, a node $v$, a candidate set $E'=\{ (u,v)| u \in N^{G_c}_{in}(v)   \}$ for insertion and the critical edge $(u^*_v,v)$, we have: $\delta^\triangle_{G_k \cup E'}(s)=\delta^\triangle_{G_k \cup \{(u^*_v,v)\}}(s)$.
\end{theorem}

\begin{proof}
Since $(u^*_v,v)$ is the critical edge, $\sppr{(s,u^*_v)}\cdot p_{(u^*_v,v)}$ is the greatest among all the incoming neighbors of $v$ in $G_c$. Based on Lemma~\ref{theorem:marginalgain}, we can know that, after inserting $(u^*_v,v)$, introducing edges sharing the same destination node as $(u^*_v,v)$ will not increase the RMPP-based influence. Thus, this theorem is deduced.
\end{proof}


Algorithm~\ref{alg:networkaugmentation} describes the process of edge insertion sketch by leveraging Theorem~\ref{theorem:critical}. Lines 1-7 construct the candidate graph $G_c$ containing the candidate edges which can be inserted into $G_k$. $G_c$ only contains the edges between nodes in $G_k$ and their neighbors in $G$, since nodes in $G_k$ cannot reach their two-hop neighbors in $G$ with only one inserted edge. Lines 8-10 compute the critical incoming neighbor for each node in $G_c$, and Lines 11-18 iteratively insert the critical edge with the maximum marginal gain into $G_k$ until all the critical edges have been inserted.

\smallskip
\noindent \textbf{Theoretical Guarantee}. Maximizing the influence spread of $s$ in the given $G_k$ under the RMPP model is already a submodular set function maximization problem with partition matroids constraint. That is, we are only allowed to pick a single edge from edges with the same destination/end node. Thus, the greedy approach (i.e., Algorithm~\ref{alg:networkaugmentation}) naturally leads to a 1/2 approximation ratio~\cite{goundan2007revisiting}. 

\smallskip
\noindent \textbf{Time Complexity}.
There are at most $O(|V|)$ iterations and each iteration takes $O(|V| (|V|+|E| )\log |V| )$ to find the critical edge with the maximum marginal gain via a variant of Dijkstra algorithm~\cite{chaoji2012recommendations}. Thus, Algorithm~\ref{alg:networkaugmentation} takes $O(|V|^2  (|V|+|E| )\log |V| )$ time. 



\begin{algorithm}[!t] 
\small
 \caption{ \method{PracticalEdgeInsertion (\innermethod{})} }\label{alg:networkaugmentationpractical}
   \SetKwInOut{Input}{Input}
   \SetKwInOut{Output}{Output}
  \Input{The input network $G=(V,E)$, the current \ksubnetwork{} $G_k=(V_k,E_k)$, the seed user $s$ and an integer $k$.}
  \Output{The updated \ksubnetwork{} $G_k$.}


$G_c=\{V_c,E_c\} \leftarrow $ constructs the candidate graph as in Algorithm~\ref{alg:networkaugmentation}\;

 $L \leftarrow$ nodes in $G_k$ excluding $s$, ranked by their subtree sizes of in the SMPP-based tree in non-increasing order\;
 $L' \leftarrow $ nodes, that are in $G_c$ but not in $G_k$, ranked by their out-degrees in non-increasing order\;
 $L \leftarrow $ Append $L'$ to the end of $L$\;

\ForEach{$v$ in $L$}{
  \lIf{$|N^{G_c}_{in}(v)|=0$}{Continue}

  $u^*_v= \argmax_{u \in N^{G_c}_{in}(v)}  \sppr{(s,u)} \cdot p_{(u,v)} $\;

  Add $(u^*_v,v)$ into $G_k$, and remove $(u^*_v,v)$ from $G_c$\;
    \If{$|N^{G_k}_{out}(u^*_{v})|=min(k,|N^{G}_{out}(u^*_{v})|)$}{
    Delete all outgoing edges of $u^*_v$ from $G_c$.
  
  }
  
  Mark all inserted edges as native\;

}
\Return $G_k$\;

\end{algorithm}

\subsubsection{Practical Edge Insertion}\label{sec:praticaledgeinsertion}

Algorithm~\ref{alg:networkaugmentation} suffers from very high time complexity and thus is infeasible in real-world large-scale social networks. In each iteration, Algorithm~\ref{alg:networkaugmentation} tries to insert a critical incoming edge to a node with the maximum marginal gain w.r.t. the influence, and correspondingly update the candidate network which can result in updates of critical edges of the remaining nodes and the marginal gain of each critical edge (Lines 11-18). We observe that the nodes considered across all iterations actually form an order. If the order is known in advance, we can directly follow this order to insert critical edges \emph{without} marginal gain computation and comparison, and hence the whole process can be notably accelerated. Unfortunately, computing this order is the main efficiency bottleneck of Algorithm~\ref{alg:networkaugmentation}. 

We find that in the special case below (i.e., Theorem~\ref{theorem:idealcase}), there is no need to know this order because, no matter what the order is, the critical edge of any remaining node will not be updated in the iterative insertion process.

\begin{theorem}\label{theorem:idealcase}
Given the original network $G$, a \ksubnetwork{} $G_k$, the initial candidate network $G_c=\{V_c,E_c\}$ (i.e., Line 7 in Algorithm~\ref{alg:networkaugmentation}), the seed node $s$, and $Cr(x)=\{ v |  \forall v \in V_c, x= u^*_v\}  $ which denotes the set of nodes whose critical neighbor is $x$. If $\forall x \in V_c, |Cr(x)| \leq min(k,N^{G}_{out}(x))- N^{G_k}_{out}(x) $, then directly inserting all critical edges leads to the optimal solution with the maximum influence. 
\end{theorem}

\begin{proof}
Since $min(k,N^{G}_{out}(x))$ describes the maximum number of outgoing edges from $x$ allowed in $G_k$,  $ N^{G_k}_{out}(x) $ describes how many of them have already been inserted into $G_k$, and $\forall x \in V_c, |Cr(x)| \\ \leq min(k,N^{G}_{out}(x)) \allowbreak - N^{G_k}_{out}(x) $, inserting all critical edges in $Cr(x)$ will not result in the critical edge/neighbor update of the rest nodes (i.e., Line 17 in Algorithm~\ref{alg:networkaugmentation}). Based on Theorem~\ref{theorem:critical}, directly inserting critical edges of all nodes leads to the maximum influence spread since replacing any critical edges with the remaining edge candidates or inserting more edges will not increase influence.  
\end{proof}

In the special case above, directly inserting all critical edges produces the optimal solution as Algorithm~\ref{alg:networkaugmentation} which does not need to update the critical edges (Line 17) in each iteration. Although this special case is rather rare in practice, it inspires us to find an order which can be easily obtained in prior and follow this order to insert an incoming edge $(u,v)$ for each node $v$ whose $\sppr{(s,u)} \cdot p_{(u,v)}$ is as large as possible. Ideally, this strategy is expected to effectively approximate Algorithm~\ref{alg:networkaugmentation}. We make the following observation to guide us to find such an order.


\smallskip
\noindent \textbf{Observation 1}. If the insertion of an edge $(u,v)$ introduces an $RMPP(s,v)$ with an influence probability greater than that of the $SMPP(s,v)$, for any node $x$ which $s$ needs to reach via $v$, the insertion of $(u,v)$ could also introduce $RMPP(s,x)$ with a greater influence probability than $SMPP(s,x)$.


Based on this observation, we should prioritize inserting incoming edges for nodes which \emph{appear} in a large number of SMPPs from $s$ to different nodes since such insertions could potentially increase the influence probabilities of the $RMPPs$ from $s$ to a large number of nodes. Based on Definition~\ref{def:smpp}, the graph formed by all SMPPs is a tree rooted at $s$, namely \emph{SMPP-based tree}, because any two nodes are connected by exactly one path. Thus, we can compute the appearance frequency by estimating the subtree size rooted at each node in the SMPP-based tree, which can be realized by a depth-first search and dynamic programming in $O(|E_k|)$ time. To help the influence spread reach a large number of users, we should also consider nodes in $G_c$ which have not been included in the \ksubnetwork{}. For these nodes, their subtree sizes are set as 0 and ranked based on their out-degrees in the input network $G$ since nodes with greater out-degrees have larger potentials to expand the \ksubnetwork{}. Algorithm~\ref{alg:networkaugmentationpractical} describes the overall process of practical edge insertion.

\begin{figure}[!t]

\centering
\includegraphics[width=0.4\textwidth]{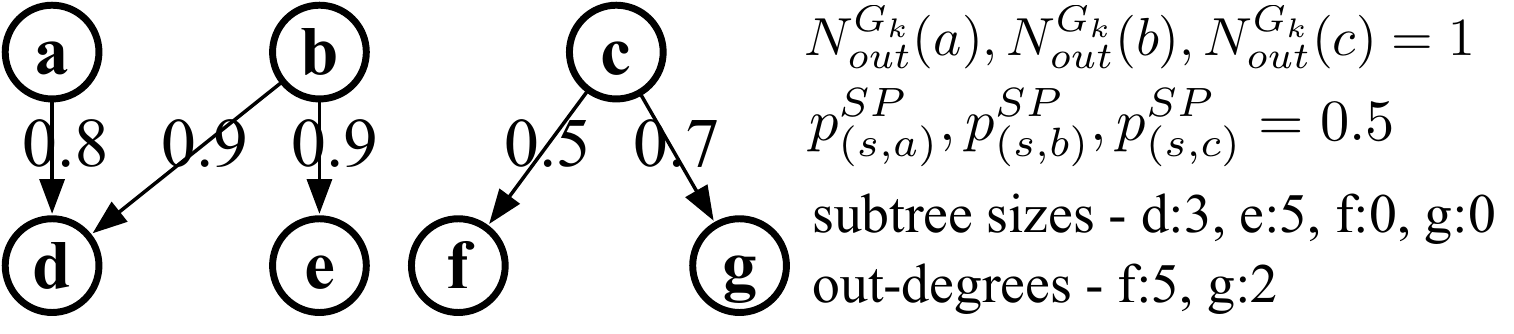}\hspace{-1em}
\caption{An example of a candidate network $G_c$ with $k=2$.} 
\label{fig:EIexample}
\end{figure}

\begin{example}\label{example:EI}
Figure~\ref{fig:EIexample} shows a candidate network $G_c$ with $k=2$. The information needed by Algorithm~\ref{alg:networkaugmentationpractical} is listed at the right side. Nodes are ordered by their subtree sizes and then out-degrees. Thus, the node ranking we follow to insert incoming edges is $e,d,f$ and $g$. When we consider node $e$, edge $(b,e)$ will be inserted into $G_k$ and edge $(b,d)$ will be deleted from $G_c$ since $N^{G_k}_{out}(b)$ will become $2=k$. When we consider $d$, edge $(a,d)$ will be inserted. Similarly, edge $(c,f)$ will be inserted but edge $(c,g)$ will be deleted when we consider node $f$ and thus there will not be incoming edge candidates for $g$.
\end{example}

\smallskip
\noindent \textbf{Time Complexity}. In Algorithm~\ref{alg:networkaugmentationpractical}, Line 1 constructs the candidate graph in $O(|V|+|E|)$ time, Lines 2-4 rank nodes based on their subtree sizes in the SMPP-based tree or degrees in $ O((|V|+|E|)\log |V|) $ time, and Lines 5-10 iteratively insert incoming edges for ranked nodes in $O(|V|+|E|)$ time since the probabilities of all SMPPs from $s$ have already been computed when we construct the SMPP-based tree in Line 2. Thus, the total time complexity is $O((|V|+|E|)\log |V|) $.


\begin{algorithm}[!t]
 \algsetup{linenosize=\footnotesize}
\small
 \caption{\method{PracticalSubnetworkAugmentation (\outtermethod{})}}\label{alg:practicalframework}
   \SetKwInOut{Input}{Input}
   \SetKwInOut{Output}{Output}
 \Input{The input network $G$, the seed user $s$, an integer $k$ and an error ratio $\epsilon$.}
  \Output{The \ksubnetwork{} $G_k$.}
  
$G_k=(V_k.E_k) \leftarrow $ \ksubnetwork{} containing the seed $s$ only\;
$G_c=(V_c,E_c) \leftarrow $ all outgoing edges of $s$ in $G$\;
$L' \leftarrow$ all outgoing neighbors of $s$ in $G$\;
 ${\delta^\triangle_{pre}}=1$\;

 \If(\tcp*[h]{\fontfamily{lmr}\selectfont \small expansion stage}){$G_c$ has edges }{ 
 
  {$G_k$= \innermethod{}$(G,G_k,G_c,s,k,L')$} 
 
  \lIf{$ (\delta^\triangle_{G_k}(s)-\delta^\triangle_{pre})/\delta^\triangle_{pre}\leq \epsilon   $}{
    Break
  }
  \lElse{$\delta^\triangle_{pre}=\delta^\triangle_{G_k}(s)$}
  
  {$G_c, L'=$ \method{UpdateCandidateGraph}$(G,G_k,G_c,s,L')$}\;

 }
 
$G_k$={ \method{FUR}}$(G,G_k,G_c,s)$\tcp*[h]{\fontfamily{cmr}\selectfont \small filling stage}

  Return $G_k$\; 
\end{algorithm}

\begin{algorithm}[!t]
 \algsetup{linenosize=\footnotesize}
\small
 \caption{\method{FillUpRecommendation (FUR)}}\label{alg:fillup}
   \SetKwInOut{Input}{Input}
   \SetKwInOut{Output}{Output}
 \Input{The input network $G$, the \ksubnetwork{} $G_k$, the candidate graph $G_c$ and the seed user $s$.}
  \Output{The updated $G_k$.}
  
  \ForEach{$v \in V_k$ where $|N_{in}^{G_k}(v) |>0$}{
     ${L^{in}[v] }\leftarrow$ a list of incoming neighbors of $v$ in $G_c$ where each neighbor $u$ is sorted by $\sppr{(s,u)} \cdot p_{(u,v)}$ in descending order\; 
    ${I[v]}=1$\;
    
  }
  
  $L \leftarrow$ a list of nodes in $V_k$ ordered by their subtree sizes in the graph formed by SMPPs from $s$ in descending order\;
  
  \While{$L$ is not empty}{
      $L' \leftarrow $ an empty list\;
      
      \ForEach{$v$ in $L$}{
      
        \ForEach{$i$ from $I[v]$ to $|L^{in}[v]|$ }{
          $u=L^{in}[v][i]$\;
          \lIf{$(u,v) \notin E_c$}{ Continue}
          \lIf{ $i \neq |L^{in}[v]|$ }{ Add $v$ into $L'$ }
          $I[v]=i+1$\; 
          Add $(u,v)$ into $G_k$; Remove $(u,v)$ from $G_c$\;
            \If{$|N^{G_k}_{out}(u)|=min(k,|N^{G}_{out}(u)|)$}{
              Delete all outgoing edges of $u$ from $G_c$.
            
            }
           Break\;
        
        }
      
      }
      
      $L=L'$\;
  
  }

  Return $G_k$\; 
\end{algorithm}

\subsection{Subnetwork Augmentation}\label{sec:augmentation}

The subnetwork augmentation sketch (i.e., Algorithm~\ref{alg:initialframework}) suffers from two issues which incur notable computation costs. 
\begin{itemize}[noitemsep,leftmargin=*]
\item Issue 1: the candidate graph in each iteration is constructed from scratch without leveraging the candidate graphs generated in previous iterations. 
\item Issue 2: the number of edge insertions in each iteration is very limited (i.e., at most one incoming edge for each node) such that it may take considerable iterations to fill up the recommendations. 
\end{itemize}

\smallskip
\noindent \textbf{Observation 2}. Regarding Issue 1, we observe that the remaining candidate graph updated by \innermethod{} (i.e., Algorithm~\ref{alg:networkaugmentationpractical}) in the last iteration is a subgraph of the input candidate graph of \innermethod{} in the current iteration. Specifically, the edges newly introduced in the current input candidate graph are outgoing edges which 1) start from nodes newly introduced in $G_k$, and 2) neither appear in $G_k$ nor the remaining $G_c$ updated by \innermethod{} in the last iteration. Based on this observation, we can simply build the candidate graph based on the previous one by only introducing these edges. Due to its simplicity and space limit, please refer to our technical
report~\cite{TR} for the pseudocode. With this update method, the process of constructing the candidate graph from scratch in \innermethod{} can be replaced. 



\smallskip
\noindent \textbf{Observation 3}. Regarding Issue 2, we observe that the influence spread of $s$ will converge after a few iterations (e.g., 6) of the subnetwork augmentation and the \ksubnetwork{} $G_k$ barely includes more nodes from $G$ after convergence (as shown in Figure~\ref{fig:noderatio} in experiments). Considering that we treat all inserted edges as native ones at the end of each iteration, the influence spread of $s$ is actually based on SMPP paths. Therefore, the converged influence also indicates that the SMPPs from $s$ to the rest of nodes barely change, and accordingly the ranking of the nodes based on their subtree sizes in the SMPP-based tree is quite stable.

 
As a result, after the convergence, we can leverage the same and converged node ranking for edge insertions in the rest of iterations of the subnetwork augmentation. Furthermore, the candidate graph at the end of the last iteration naturally becomes the initial candidate graph in the current iteration since no more new nodes are considered. With these properties, we propose a method called \method{FillUpRecommendation (FUR)} (Algorithm~\ref{alg:fillup}) which simulates the process of \innermethod{} in the rest of all iterations of the subnetwork augmentation. Specifically, each while loop from Line 6 to 17 corresponds to one iteration where we greedily insert the critical edge for each node in $L$ constructed based on the converged ranking. In each loop, $L'$ is used to store nodes which still have incoming neighbors and thus will be considered in the next iteration, and $I[v]$ is used to efficiently retrieve the critical neighbor by pruning out neighbors which either were considered in previous iterations or cannot be connected to $v$ anymore due to the update of $G_c$ in previous iterations (Lines 14-15). Since each incoming neighbor has been sorted in Line 2, simulating each iteration is very efficient with aforementioned data structures.

\begin{figure}[!t]

\centering
\includegraphics[width=0.4\textwidth]{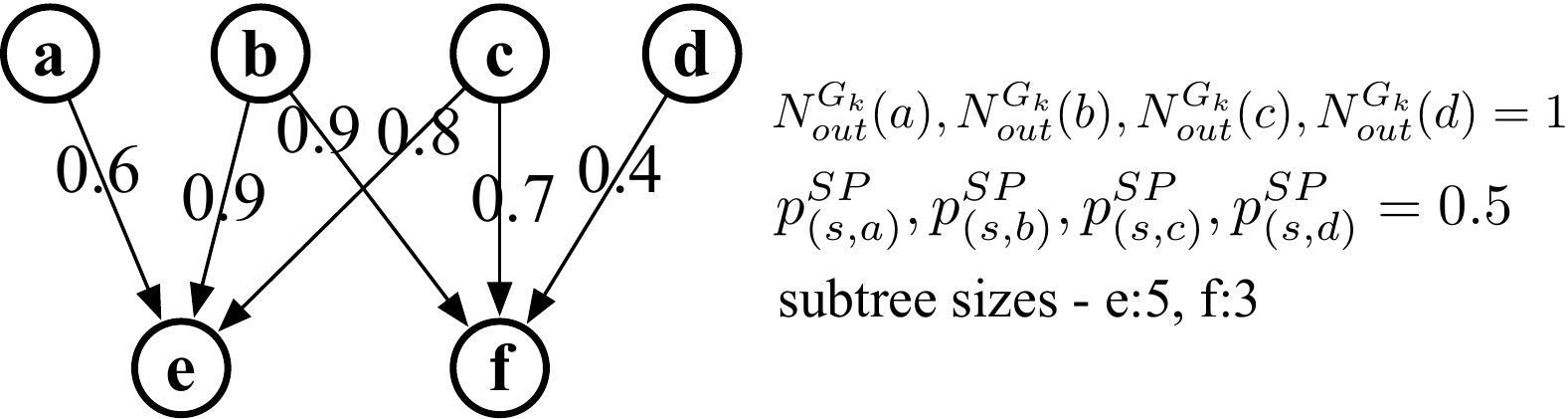}\hspace{-1em}
\caption{An example of a candidate network $G_c$ with $k=2$.} 
\label{fig:FURexample}
\end{figure}



\begin{example}\label{example:FUR}
Figure~\ref{fig:FURexample} shows a candidate network $G_c$ with $k=2$. Suppose that the influence spread in $G_k$has converged at iteration $i$ and the converged node ranking is $e$ and $f$. Algorithm~\ref{alg:networkaugmentationpractical} will be executed two more iterations to reach the termination condition (i.e., no edges exist in $G_c$). At iteration $i+1$, edges $(b,e)$, $(b,f)$, $(c,f)$ and $(c,e)$ will be inserted, deleted, inserted and deleted in sequence respectively. The remaining candidate graph becomes the initial candidate graph at iteration $i+2$ where edges $(a,e)$ and $(d,f)$ will be inserted. The method \method{FUR} (Algorithm~\ref{alg:fillup}) is proposed to simulate the iterative process after the convergence and the while loop in \method{FUR} will be executed twice.
\end{example}

\smallskip

\noindent \textbf{Practical Subnetwork Augmentation}. By incorporating the measures for tackling the two issues above, we propose a framework called \method{PracticalSubnetworkAugmentation (\outtermethod{})} (Algorithm~\ref{alg:practicalframework}) where the predefined $\epsilon$ (e.g., $10^{-4}$) controls the convergence point (Line 7) and \method{FUR}, based on the converged subtree-size-based node ranking (Line 10), repeats the process of Lines 5-6 in \outtermethod{} until no edges exist in $G_c$. Note that $L'$ records the new candidate nodes to be incorporated into $G_k$ for graph expansion in the last iteration and is updated by \method{UpdateCandidateGraph} in Line 9. For ease of illustration in our experiments later, we regard the process before convergence as the \emph{expansion stage} (i.e., Lines 5-9) and the process afterwards as the \emph{filling stage} (i.e., Line 10).



\smallskip
\noindent \textbf{Time Complexity}. Suppose Algorithm~\ref{alg:practicalframework} takes $I$ iterations to converge. Lines 5-9 take $O(I(|V|+|E|)\log |V| )$ time which is acceptable in practice since $I$ is empirically small (e.g., 6). Lines 1-3 in Algorithm~\ref{alg:fillup} take $O(|V|d_{m}\log d_{m} )$ where $d_m$ refers to the maximum in-degree in $G$, and Lines 5-17 take $O(|E|)$ time as each incoming neighbor of each node is visited only once. Thus, the total time complexity of Algorithm~\ref{alg:practicalframework} is $O(I(|V|+|E|)\log |V|+|V|d_{m}\log d_{m} )$.


%% file: tex/methodformultiple.tex
\section{Solving \short{} for multiple seeds}\label{sec:methodmultiple}
When there is a set $S$ of independent seed nodes, directly adopting the previous idea of inserting at most one incoming edge for each node is not feasible, since the SMPPs from seed nodes to the same node can be different and hence the critical edge for each node is seed-specific. The most straightfoward approach to handle the case is to iteratively insert the edge, which brings the maximum marginal gain to the sum of the influence increment of all seed nodes, until no edges can be inserted into the \ksubnetwork{}. However, this approach is impractical due to expensive marginal gain computation and a huge number of iterations needed to converge.  

\smallskip
\noindent \textbf{Observation 4}. We observe that an edge $(u,v)$ insertion can have different levels of impact on increasing the influence spread of different seed nodes. For example, if $u$ is a neighbor of the seed node $s_1$ but ten hops away from the seed node $s_2$, inserting $(u,v)$ is more likely to bring more influence to $s_1$ than $s_2$. Furthermore, as the iteration goes, the \ksubnetwork{} will be expanded with more topological information during the augmentation process. Thus, we should not `waste' many unnecessary recommendation opportunities in the current stage of the \ksubnetwork{} since they can be saved for better decision making in later iterations. 

Thus, to make the best use of candidate recommendations, we should  focus only on the recommendation which will bring large influence increment to the relevant seed node. Specifically, we define $r(u)=\argmax_{s\in S} \sppr{(s,u)}$ as the \emph{relevant seed node} of $u$ since inserting outgoing edges from $u$ would be more likely to increase the influence of the seed node $r(u)$. For all incoming neighbors of node $v$, we should connect the critical neighbor $u^*_v$ to $v$ which forms the RMPP with the maximum probability from the corresponding seed node $r(u^*_v)$, i.e., $u^*_v=\argmax_{ u\in N_{in}^{G_c}(v) }   \sppr{(r(u),u)} \cdot p_{(u,v)} $. 

Nodes with the same relevant seed node $s$ and SMPPs from $s$ to them naturally form an SMPP-based tree. We can construct the SMPP-based trees based on each seed node and these trees are disjoint (i.e., no node overlap). With these trees, we can compute the subtree size of each node in the corresponding tree and prioritize inserting the nodes with large subtree sizes, because increasing the probability to influence the node $u$ will also increase the probability to influence the descendants of $u$ in the SMPP-based trees.

\begin{figure}[t]
\centering
\subfloat[iteration $i$.]{
\includegraphics[width=0.2\textwidth]{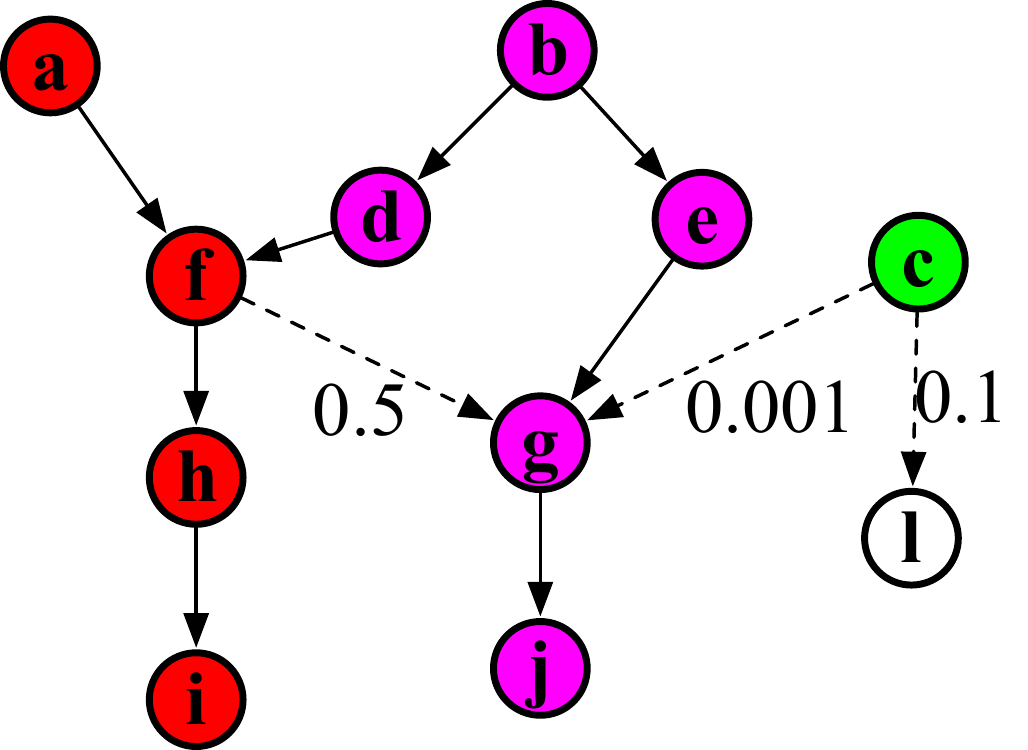}\hspace{-1em}
} 
\hspace{0.5mm}
\subfloat[iteration $i+1$.]{
\includegraphics[width=0.2\textwidth]{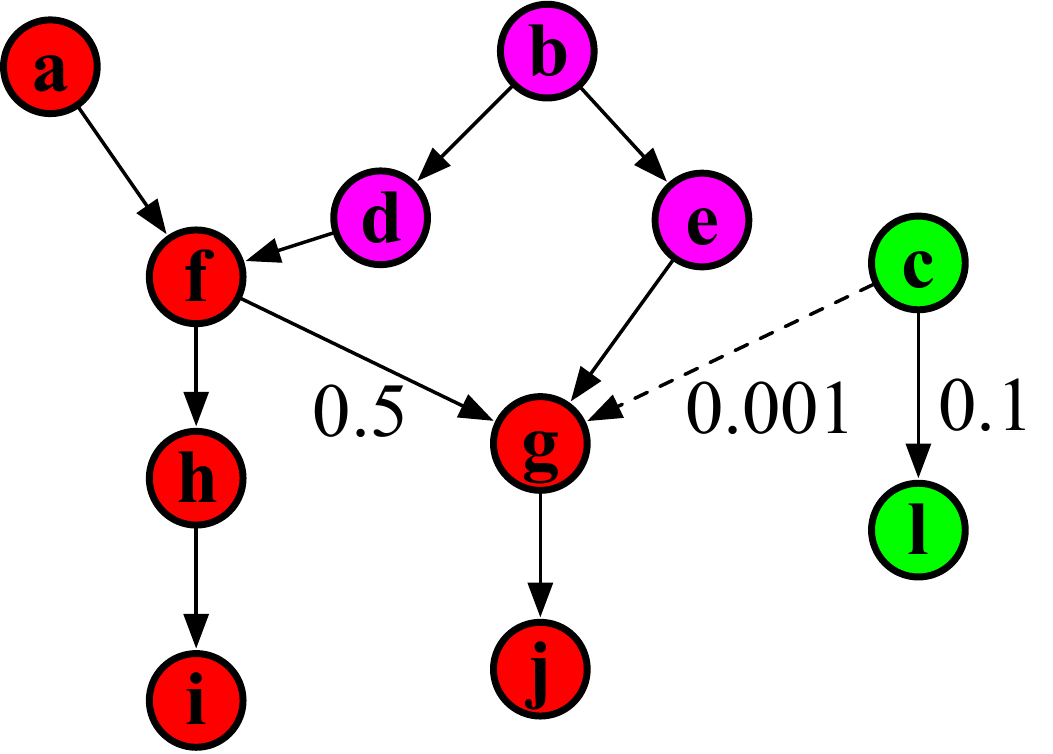}\hspace{-1em}
}

\caption{SMPP-based subgraphs in two iterations.}\label{eg:subgraph}
\label{comparison}

\end{figure}

\begin{example}
Figure~\ref{eg:subgraph} shows an example of the SMPP-based trees of a 2-subnetwork in two consecutive iterations of the \outtermethod{}. Here, $a$, $b$ and $c$ are the seed nodes, all solid edges have a weight of 0.1, nodes with the same color belong to the same SMPP-based trees, and white nodes are candidates for the 2-subnetwork expansion. In iteration $i$, node $g$ is processed before $l$ because $g$ has a larger subtree size (i.e., 2).
We have two candidate incoming edges to node $g$, and $(f,g)$ will be inserted since $f=\argmax_{ u\in \{f,c\} }   \sppr{(r(u),u)} \cdot p_{(u,g)} $ where $r(f)=a$ and $r(c)=c$. In the input 2-subnetwork of iteration $i+1$, the relevant seed nodes of $g$ and $j$ are updated to node $a$ since inserting outgoing edges from them is more likely to increase the influence of $a$. 
\end{example}

The premise of the aforementioned idea requires computing the SMPP between each node $u$ and $r(u)$. A straightforward approach is to enumerate SMPPs from every seed node to $u$ and get the one with the greatest probability. However, this is very expensive in real-world large-scale social networks, especially when $|S|$ is large. To mitigate this issue, we can simply introduce edges with the same influence probability of 1 from a \emph{virtual} node $x$ to every seed node. Afterwards, we only need to adopt a variant of Dijkstra algorithm~\cite{chaoji2012recommendations} to compute the SMPPs from $x$ to the rest of the nodes; the second node in each SMPP must be a seed node, and be the relevant seed node of the end node of this path. With the virtual node $x$, we simplify the case of multiple seed nodes into the case of a single seed node $x$ since edge insertion can be guided by the probability of SMPP from $x$ instead of relevant seed nodes, i.e., 
$$u^*_v=\argmax_{ u\in N_{in}^{G_c}(v) }  1\cdot   \sppr{(r(u),u)} \cdot p_{(u,v)}=\argmax_{ u\in N_{in}^{G_c}(v) }   \sppr{(x,u)} \cdot p_{(u,v)}$$
 Since we transform the \short{} problem for multiple seed nodes into the one for a single virtual node $x$, we can directly adopt \outtermethod{} for the single seed node with minor adjustments: (1) introduce $O(|S|)$ edges with the same influence probability 1 from a virtual node $x$ to every seed node in the initial \ksubnetwork{}; (2) remove the virtual node and edges between it and seed nodes from the output \ksubnetwork{} of \outtermethod{}. 

 \noindent\textbf{Time Complexity}.~Considering that we only change the topology of the \ksubnetwork{}, the previous time complexity analysis for the single case still applies. Thus, the total time complexity is $O(I(|V|+|E|+|S|)\log |V|+|V|d_{m}\log d_{m} )$.

%% file: tex/shortenedexp.tex
\section{Experiment}\label{sec:experiment}
{In this section, we will conduct experiments on three problems to demonstrate the robustness and effectiveness of our methods:  }

\begin{itemize}[leftmargin = *]

\item The first is our proposed problem, \short{},  where the extensions of the methods in the first problem are compared (Section~\ref{sec:exp-closed}). Note that the core difference between these two problems has been illustrated and please refer to Section~\ref{sec:relatedwork} for details.

\item The second problem is maximizing users' Click-trough Rate in an activity of an online Tencent application, where we deploy our method and evaluate how it helps improve user retentions and interactions (Section~\ref{sec:deploy}). 

\item The third problem is maximizing the influence in \emph{open} social networks via edge insertions, where existing baselines including the state-of-the-art~\cite{chaoji2012recommendations} are compared. The results show that, compared with~\cite{chaoji2012recommendations}, our method \outtermethod{} achieves very competitive results with up to five-orders-of-magnitude speedup. Please refer to the technical report~\cite{TR} for details.

\end{itemize}

\smallskip
\noindent\textbf{Datasets}. Table~\ref{table:dataset} presents all the real-world undirected social networks used. In particular, MOBA and MOBAX correspond to two friendship networks of Tencent multiplayer online battle arena games, RPG corresponds to a friendship network of a role-playing game, and the other datasets are available in~\cite{konect}. Note that each edge is represented twice since the influence propagation is directed. The first four datasets and the last dataset will be used for the second and the third problem respectively.


\smallskip
\noindent \textbf{Environments}. We conduct all experiments on a Linux server with Intel Xeon E5
(2.60 GHz) CPUs and 512 GB RAM. All algorithms are implemented in
Python and our code is available at~\cite{TR}. 

\begin{table}[]
\caption{Dataset Statistics}\label{table:dataset}
\small
\begin{tabular}{cccc}
\hline
\textbf{Dataset} & \textbf{|V|} & \textbf{|E|} & \textbf{Avg. Degree} \\ \hline
Catster &149,700 & 10,898,550 & 73\\
MOBA & 503,029 & 9,372,022 & 19 \\
RPG  & 2,331,047 & 88,227,562 & 38\\
Orkut & 3,072,441 & 234,369,798 & 76\\
MOBAX & 36,201,207&3,281,207,036&90\\ \hline
\end{tabular}
\end{table}

\begin{table}[]
\caption{Effectiveness comparison with $k=30$ and $|S|=50$.}\label{table:orgvsbst}
\small
\begin{tabular}{cccccccc}
\hline
\multirow{2}{*}{\textbf{Dataset}} & \multicolumn{2}{c}{\textbf{Degree}} & \multicolumn{2}{c}{\textbf{FoF}} & \multicolumn{2}{c}{\textbf{Random}} & \multirow{2}{*}{\textbf{\outtermethod{}}} \\ \cline{2-7}
 & \textbf{Ori} & \textbf{Bst} & \textbf{Ori} & \textbf{Bst} & \textbf{Ori} & \textbf{Bst} &  \\ \hline
Catster & 3.1E2 & 4.2E5 & 3.8E2 & 4.8E5 & 3.6E3 & 6.3E5 & \textbf{8.6E5} \\
MOBA & 5.0E6 & 7.4E6 & 6.5E6 & 7.7E6 & 7.0E6 & 8.0E6 & \textbf{9.8E6} \\
RPG & 1.1E6 &2.1E7 &   2.4E6 &2.2E7 &1.9E7 & 2.4E7 & \textbf{3.2E7} \\
Orkut & 7.0E2 & 5.4E7 & 6.8E2 & 5.6E7 & 9.6E5 & 5.9E7 & \textbf{8.2E7} \\ \hline
\end{tabular}

\end{table}

\begin{table}[]
\caption{Running time (s) with $k=30$ and $|S|=50$ where the numbers in brackets under \outtermethod{} refers to the time cost in the graph expansion stage. Note that the number of common friends is pre-computed in \method{Bst-FoF}.}\label{table:orgvsbsttime}
\small
\begin{tabular}{cccccccc}
\hline
\textbf{Dataset} & \textbf{Bst-Degree} & \textbf{Bst-FoF} & \textbf{Bst-Random} & \textbf{\outtermethod{}} \\ \hline
Catster & 1.1E2 & 1.0E2 & 9.8E1 & 1.2E2 (9.1E1)\\
MOBA &1.4E2& 1.3E2 & 1.2E2 & 1.8E2 (1.3E2)\\
RPG & 9.6E2& 9.1E2 & 8.8E2 & 9.8E2 (7.3E2)\\
Orkut & 4.4E3& 4.2E3& 4.1E3 & 5.1E3 (3.6E3)\\ \hline
\end{tabular}

\end{table}

\subsection{Experiment on the \short{} problem}\label{sec:exp-closed}
 We conduct seven experiments to demonstrate that: (1) the extension of existing baselines on open social networks cannot well address our problem but their performance can be significantly boosted with the help of our method \outtermethod{} (Exp1); (2) \outtermethod{} consistently and significantly outperforms all the boosted baselines with different $|S|$ and $k$ (Exp2-3); (3) how $k$ and $|S|$ impact the convergence of \outtermethod{} (Exp4); (4) \outtermethod{} is highly scalable to handle large-scale datasets (Exp5); (5) the influence convergence is also a good indicator of the node size convergence in the diffusion network (Exp6); (6) the influence of seeds in the diffusion network, produced by \outtermethod{} with a small edge size, is very competitive with the influnce of the seeds in the original network (Exp7).

\begin{figure}[t]
\includegraphics[width=0.4\textwidth]{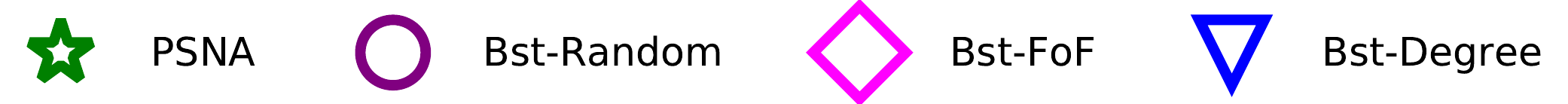}
\vspace{-0.25cm}

\centering
\subfloat[Catster]{
\includegraphics[width=0.124\textwidth]{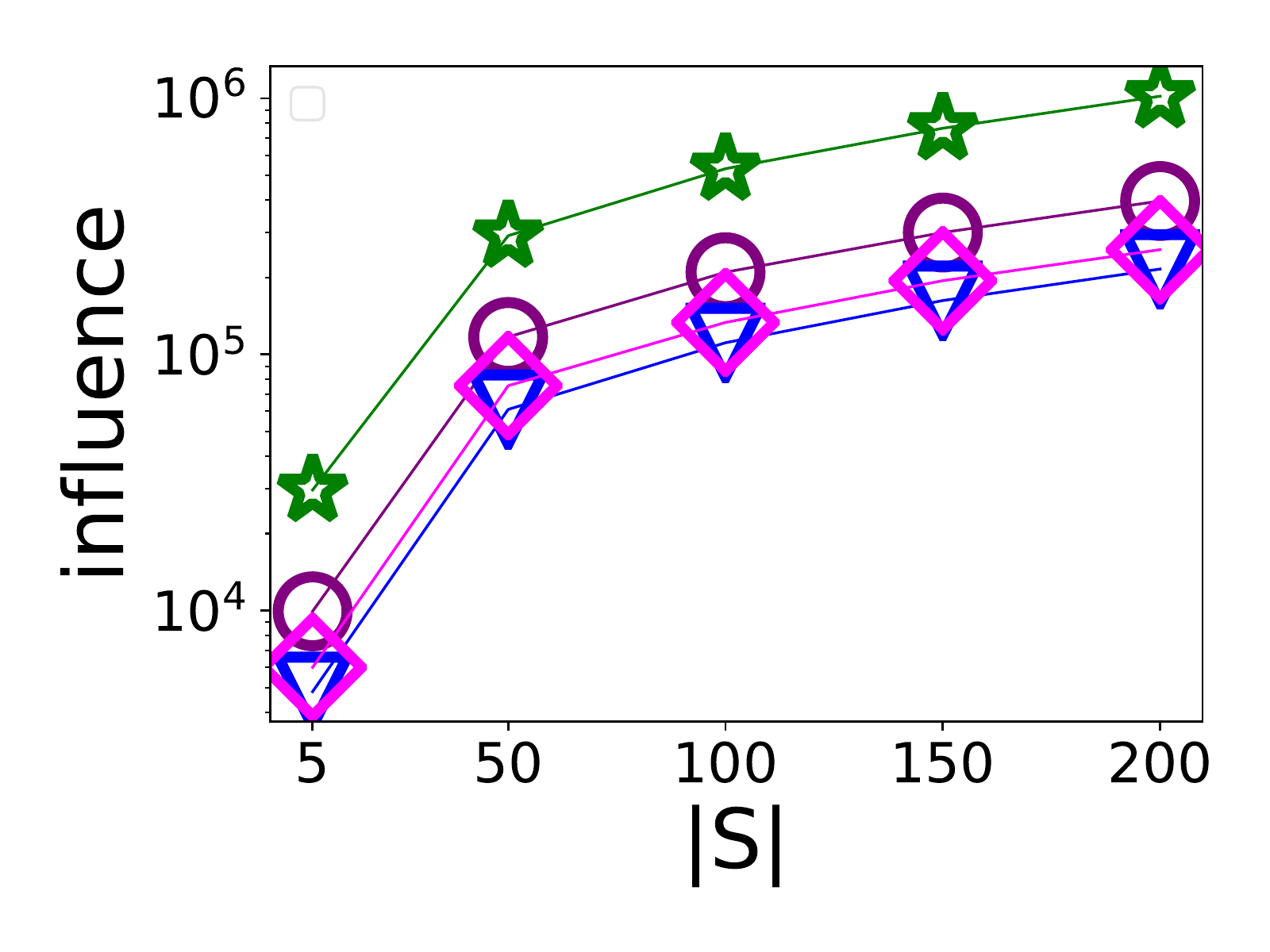}\hspace{-1em}
}
\hspace{0.1mm}
\subfloat[MOBA]{
\includegraphics[width=0.124\textwidth]{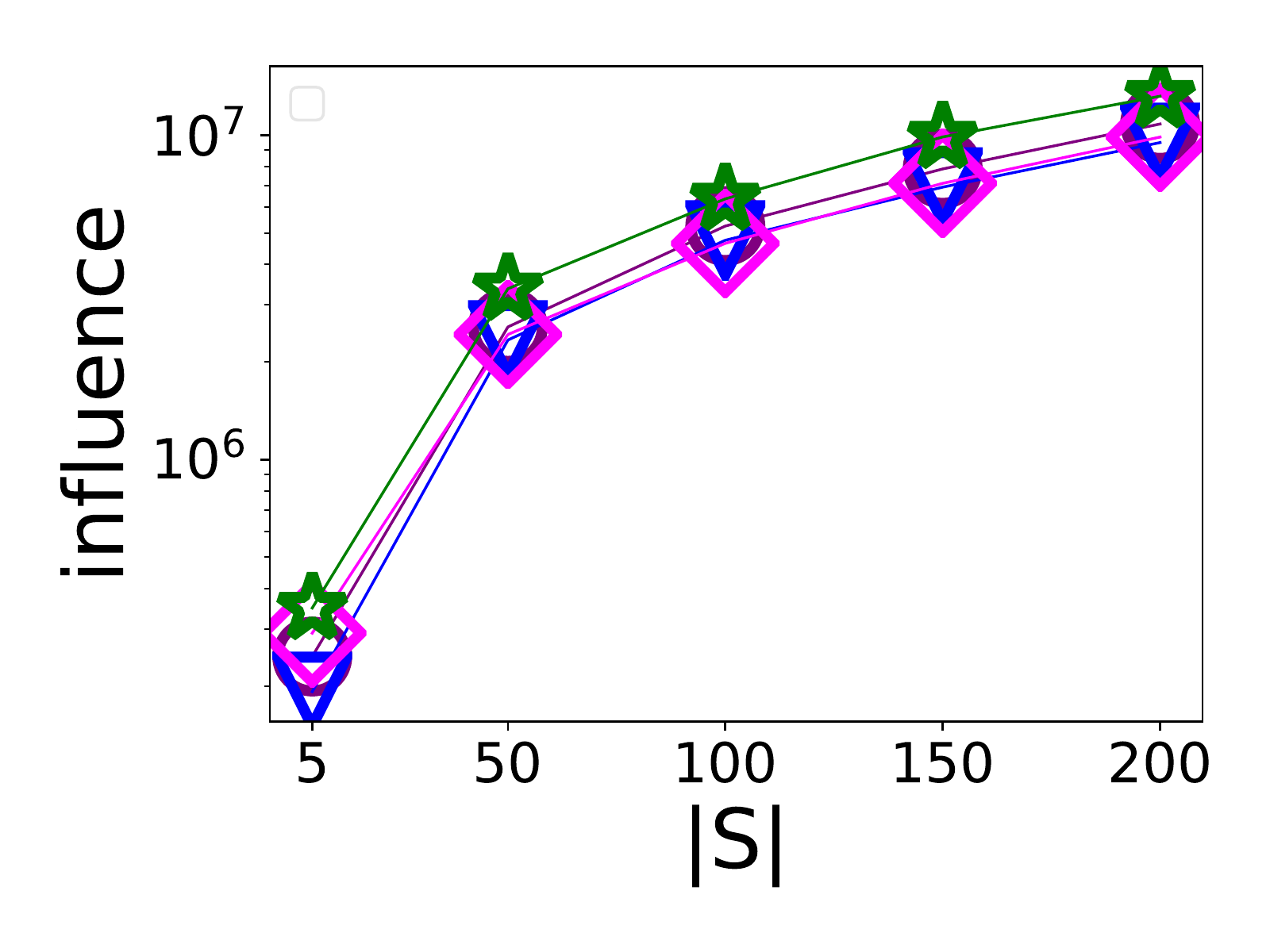}\hspace{-1em}
} 
\hspace{0.1mm}
\subfloat[RPG]{
\includegraphics[width=0.124\textwidth]{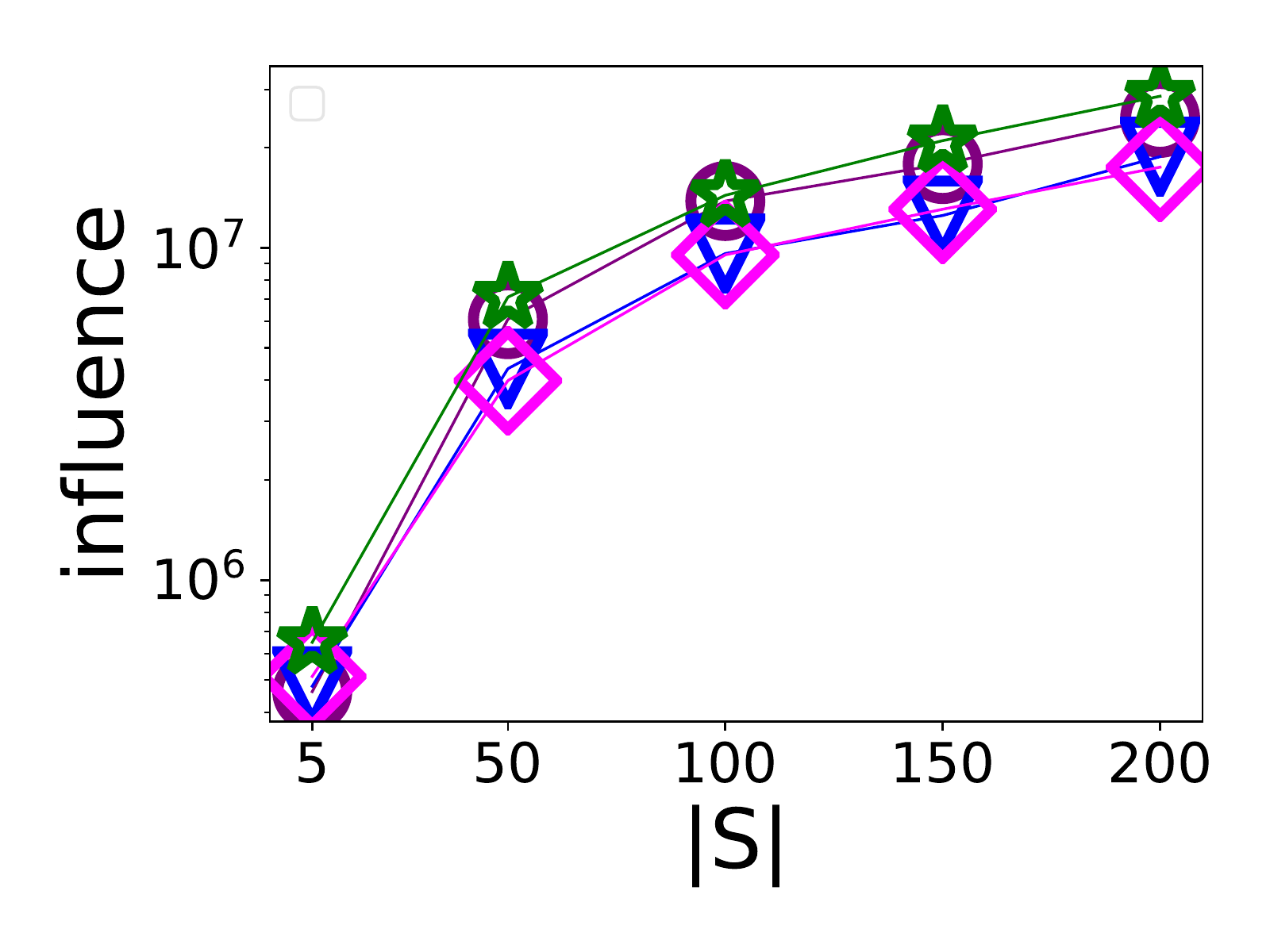}\hspace{-1em}
}
\hspace{0.1mm}
\subfloat[Orkut]{
\includegraphics[width=0.124\textwidth]{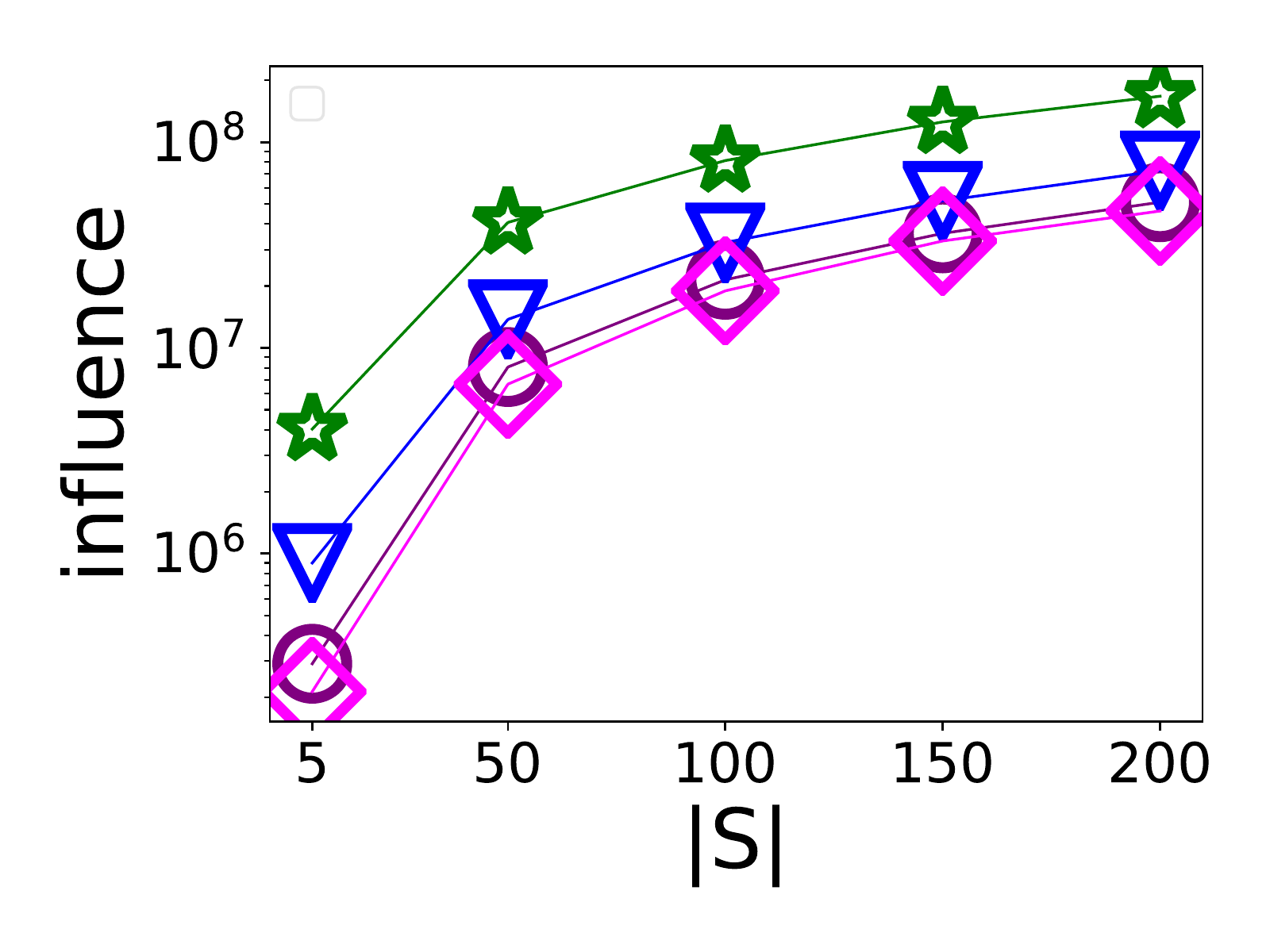}\hspace{-1em}
}

\vspace{-0.7em}
\caption{Performance Comparison with different $|S|$.}\label{fig:diffs}
\end{figure}

\begin{figure}[t]
\includegraphics[width=0.4\textwidth]{./drawing2/legend1}
\vspace{-0.25cm}

\centering

\subfloat[Catster]{
\includegraphics[width=0.124\textwidth]{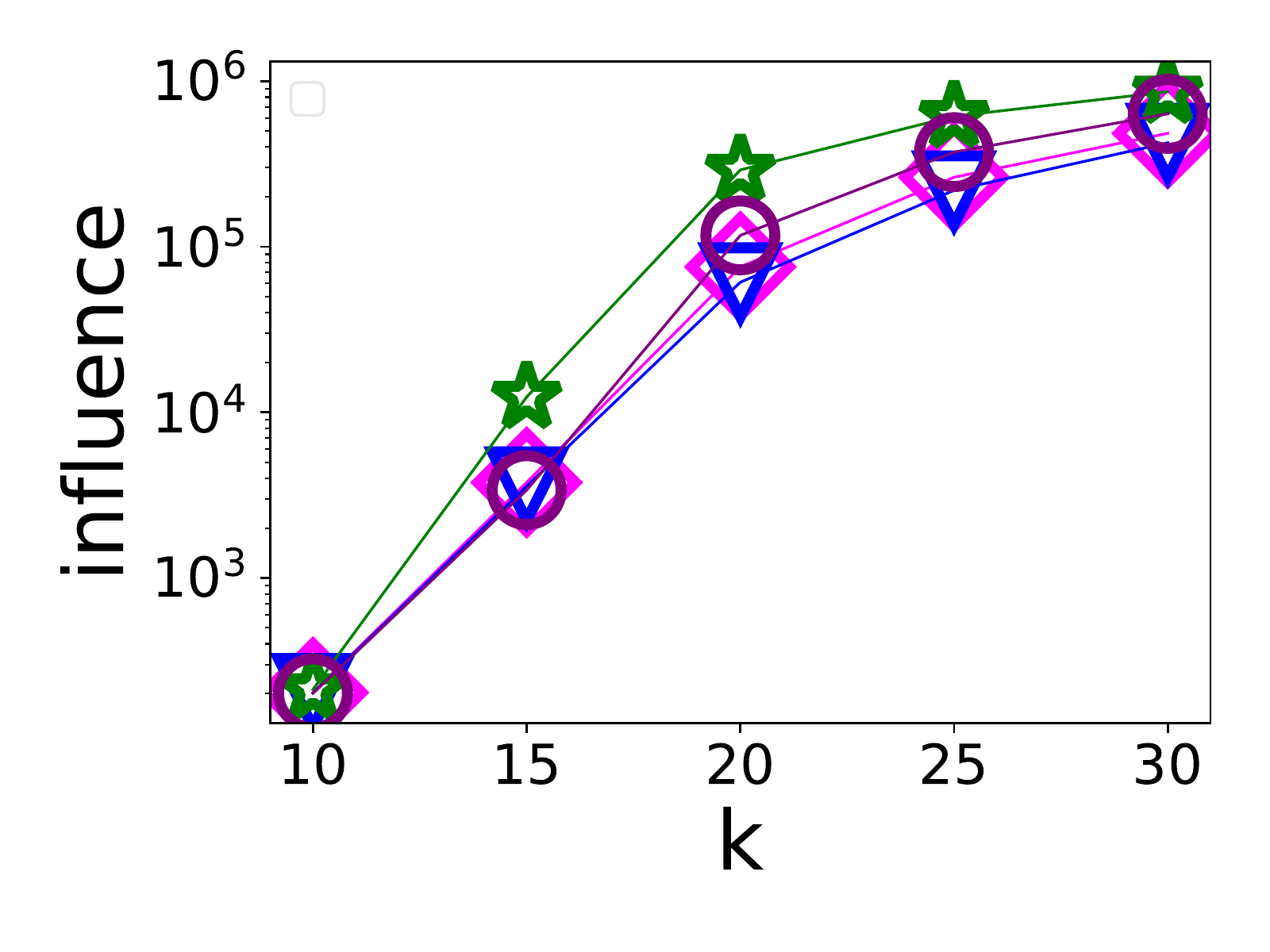}\hspace{-1em}
}
\hspace{0.1mm}
\subfloat[MOBA]{
\includegraphics[width=0.124\textwidth]{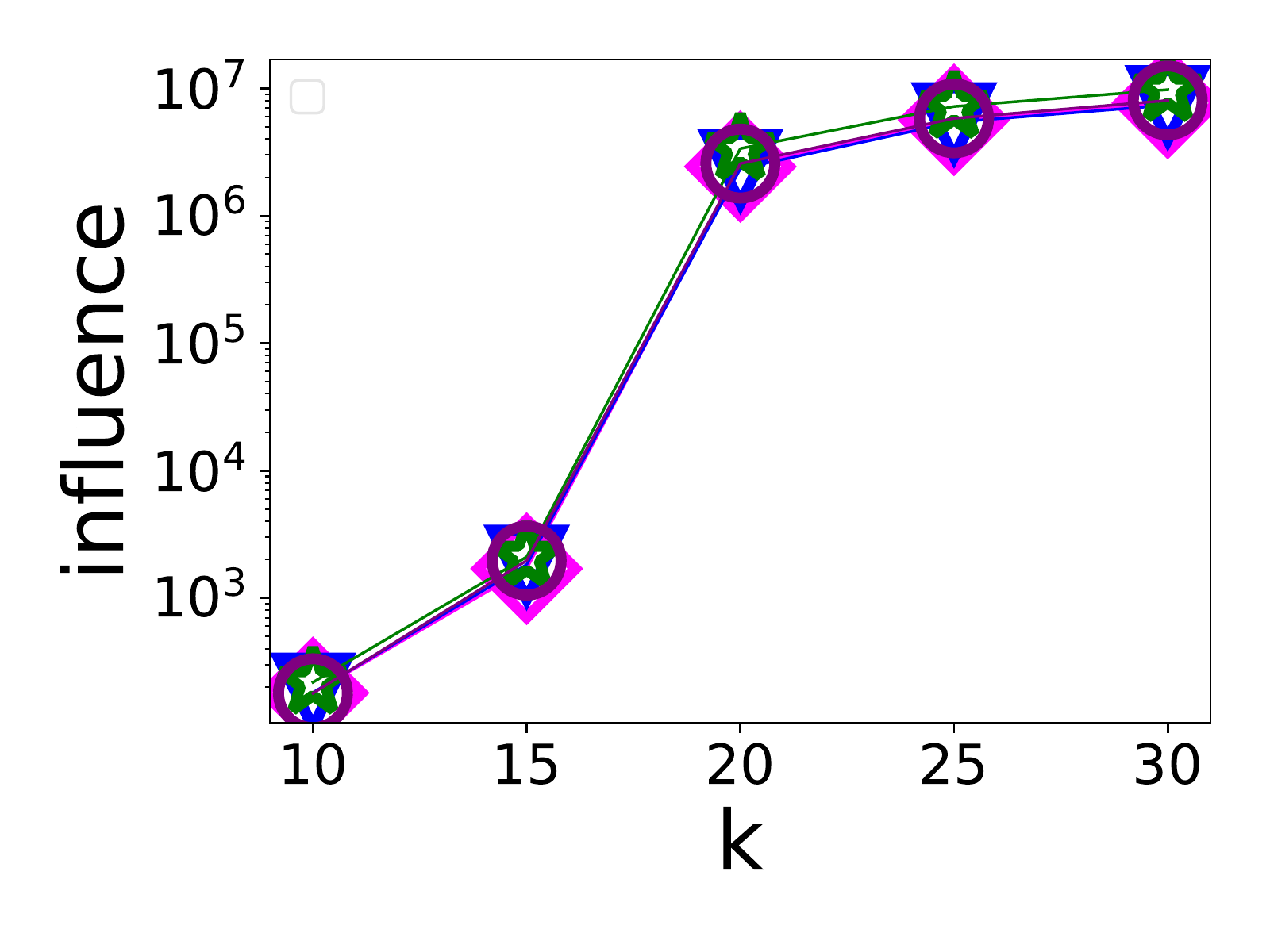}\hspace{-1em}
} 
\hspace{0.1mm}
\subfloat[RPG]{
\includegraphics[width=0.124\textwidth]{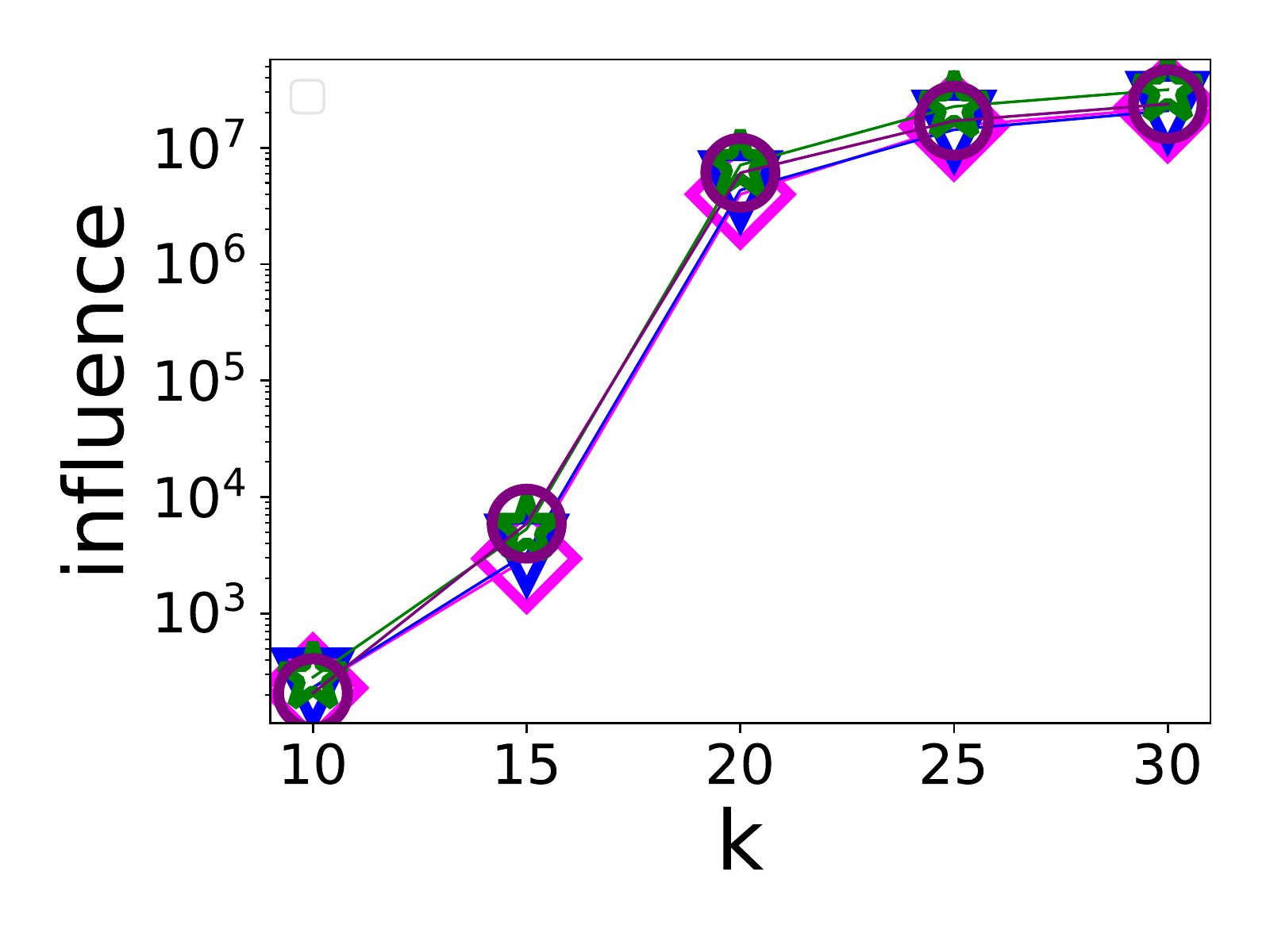}\hspace{-1em}
}
\hspace{0.1mm}
\subfloat[Orkut]{
\includegraphics[width=0.124\textwidth]{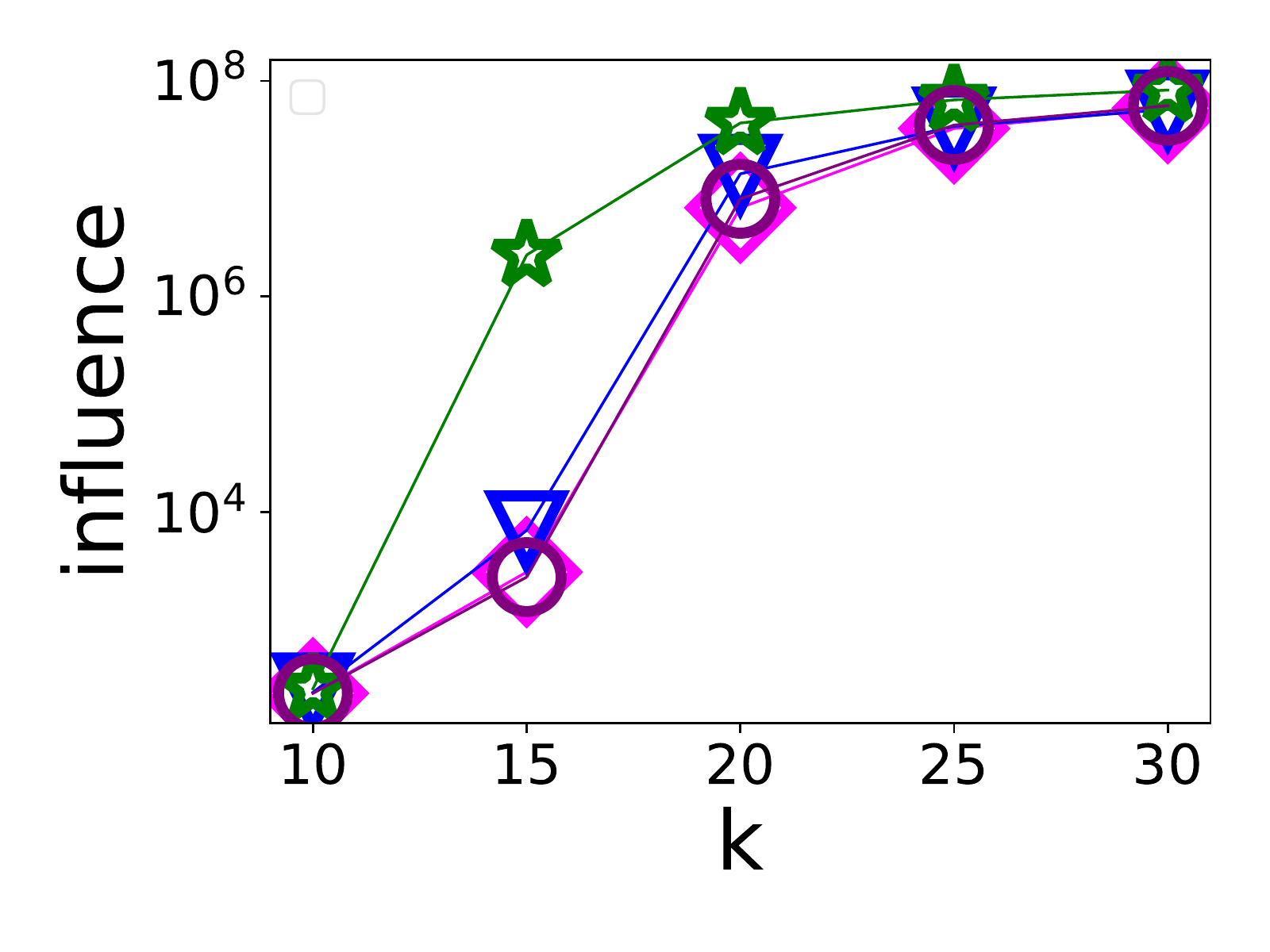}\hspace{-1em}
}

\vspace{-0.7em}
\caption{Performance Comparison with different $k$.}\label{fig:diffk}

\end{figure}

\smallskip
\noindent \textbf{Baselines}. We compare three baselines (listed below) on open social networks and their boosted versions. Thus, in total, there are six baselines. For ease of description, we assume out-degrees of nodes in the original network to be greater than $k$.

\begin{itemize}[leftmargin = *]
\item \method{Degree} based method where, for candidates sharing the same source node, top $k$ edges with the highest degrees of end nodes are selected.

\item \method{FoF} based method where,  for candidates sharing the same source node, top $k$ edges whose source and end nodes share the highest number of common friends are selected.

\item \method{Random} based method where $k$ edges are randomly selected from candidates sharing the same source node. Its performance is reported as the average over five independent runs.

\end{itemize}

Since each baseline has two versions, original and boosted, we use the prefixes `Ori' and `Bst' to distinguish them. The boosted versions are built upon the diffusion network generated by the \emph{expansion} stage of \outtermethod{} (i.e., Lines 5-9 in Algorithm~\ref{alg:practicalframework}) and work similar to the original ones. The only difference is that, in the boosted version, some important edges have been inserted by \outtermethod{} and these baselines only `fill up' recommendations for those nodes in the generated diffusion network. Hence, if we assume that the out-degrees of nodes in the original network are greater than $k$ and our method has already inserted $h <k$ outgoing edges from $u$, boosted versions just need to select $k-h$ outgoing edges for $u$.

\smallskip
\noindent \textbf{Edge Weight Settings}. Considering that the influence probability of edges are different in practice, we adopt the commonly used trivalency model~\cite{chen2010scalable} for Catster and Orkut without edge attributes. It randomly assigns a weight for each edge, from $\{10^{-1},10^{-2},10^{-3}\}$. 

For MOBA and RPG, we assign edge weights based on the intimacy of friendship. In MOBA and RPG, each pair of friends has different levels of intimacy which describes the number of interactions (e.g., the number of games they play together, the number of gifts sent from one to another). Since the pair-wise intimacy in both MOBA and RPG is represented as integers, here we transform it into an influence probability within [0,1] for deployment in the IC model. According to the Susceptible-Infected-Recovered (SIR) model~\cite{anderson1992infectious} and heterogeneous mean-filed theory~\cite{newman2002spread,cohen2011resilience,castellano2010thresholds,xie2021detecting}, the lowest influence probability should not be smaller than a constant $\lambda$ times $\beta_c= \sum_{v \in V} |N_{out}(v)| / (\sum_{v \in V} |N_{out}(v)|^2- \sum_{v \in V} |N_{out}(v)|)$, where $\beta_c$ is the epidemic threshold in the SIR model and calculated as 0.024 and 0.001 on MOBA and RPG, respectively. If $\lambda \beta_c$ is too small, the influence of the seeds will be quite limited. If $\lambda \beta_c$ is too large, the influence spread can cover a large percentage of nodes, irrespective of where it originated, and the methods' performance cannot be well compared. By following existing studies~\cite{shu2015numerical,lu2016h}, we
determine $\lambda$ by simulation on real networks. Specifically, we determine $\lambda_1$ and $\lambda_2$ and control the edge weights within the range $[\lambda_1\beta_c, \lambda_2\beta_c]$. Each influence probability $p_{(u,v)}= ((u,v)_{I} - \min_{e\in E} e_{I}) / (\max_{e\in E} e_{I} -\min_{e\in E} e_{I}  ) (\lambda_2 -\lambda_1)\beta_c +\lambda_1\beta_c$ where $(u,v)_{I}$ denotes the intimacy between $u$ and $v$. Note that we have checked several settings of $\lambda_1$ and $\lambda_2$ and these different settings will not affect the conclusion (i.e., performance ranking of methods). To test these methods' robustness to the influence probability distributions, in both MOBA and RPG, we set the range $[\lambda_1\beta_c, \lambda_2\beta_c]$ as $[0.007,0.01]$ which has a dramatically different distribution from the trivalency model deployed for other datasets.

\smallskip
\noindent \textbf{Parameter Settings}. We randomly select $|S|$ (50 by default) nodes from the top 1\% nodes with the highest degrees in the original network as the seed nodes, set $k=20$ and the error ratio $\epsilon=10^{-4}$ by default. We evaluate solution quality based on the IC-based influence spread via 10,000 Monte Carlo simulations.

\smallskip
\noindent \textit{\textbf{Exp1 - Case study on the two versions of baselines}}. Table~\ref{table:orgvsbst} and Table~\ref{table:orgvsbsttime} compare the effectiveness and efficiency with $k=30$ and $|S|=50$. We have four main observations:
\vspace{-1mm}
\begin{enumerate}[leftmargin = *]
\item \method{Org-Random} is more effective than the original versions of other baselines. The reason is that nodes with high degrees will `attract' much more incoming edges in other strategies (i.e., \method{Org-Degree} and \method{Org-FoF}) which `waste' a lot of recommendation opportunities to repeatedly influence/activate these nodes. 

\item The boosted versions can achieve about five-orders-of-magnitude larger influence than their original counterpart but can still be notably outperformed by \outtermethod{}. Specifically, \outtermethod{} can still outperform the second best performer by at least 23\%-39\% on different datasets respectively, which demonstrates the effectiveness of both two stages (i.e., the expansion and filling stages). 

\item We also compute the constitution of the edges recommended by \outtermethod{} in the boosted baselines, and the result shows that \outtermethod{} only contributes 17\%-27\% of total edges in the diffusion network. The significant performance improvement of baselines with limited involvement of \outtermethod{} further demonstrates the effectiveness of \outtermethod{} in terms of identifying important edges and nodes for increasing the seeds' influence. 

\item \outtermethod{} is very competitive with other boosted baselines in terms of running time, because all methods involve the \emph{expansion} stage of \outtermethod{} (i.e., Lines 5-9 in Algorithm~\ref{alg:practicalframework}) which dominates the total computational cost.
\end{enumerate}
\vspace{-1mm}

Due to the poor performance of the original baselines, we only use the boosted versions for comparison in the rest experiments.

\smallskip
\noindent \textit{\textbf{Exp2 - Effectiveness comparison with different $|S|$}}. 
Figure~\ref{fig:diffs} compares the performance with different $|S|$. Our method \outtermethod{} consistently outperforms other methods across all instances while the performance of other methods is not stable, e.g., \method{Bst-Random} outperforms \method{Bst-Degree} on Catster but the latter is more effective on Orkut. Another interesting observation is that the performance ranking of methods on the same dataset is consistent across different $|S|$. We think that a good edge recommendation for a non-seed user $u$ under an instance with a small $|S|$ can also be an effective recommendation for $u$ as the seed user under an instance with a great $|S|$. Thus, the diffusion networks generated by the expansion stage of \outtermethod{} under different $S$ can be similar such that edges chosen by a specific baseline upon these networks have large overlaps, which explains this observation.

\smallskip
\noindent \textit{\textbf{Exp3 - Effectiveness comparison with different $k$}}. 
Figure~\ref{fig:diffk} compares the performance at different $k$ across all datasets. We have four main observations: 
\vspace{-1mm}
\begin{enumerate}[leftmargin = *]
\item \outtermethod{} outperforms the boosted baselines and can achieve up to two-orders-of-magnitude larger influence (e.g., on Orkut). 
\item The performance gap becomes smaller when $k$ is larger. The reason is that all methods make almost the same recommendation strategy to low-degree nodes since their out-degrees are close to or even smaller than $k$. In this case, the number of combinations of $k$ outgoing edges of these nodes is very limited. This observation also explains why the performance gap is larger on datasets with greater average degrees. 
\item When $k$ is small (e.g., 10), all methods achieve similar influence since it is barely possible for any method to achieve large influence with very limited edges. 
\item The performance gap on MOBA and RPG is not as significant as that on other datasets. That is because the average edge weight on MOBA and RPG are notably larger. Thus, regardless of how the diffusion network is generated, seeds can easily influence many nodes with sufficient influence probabilities. 
\end{enumerate}
\vspace{-1mm}
      
Note that the performance difference may not be easily distinguished visually due to the log scale of y axis, but we have shown the notable performance difference with $k=20$ and $k=30$ in Figure~\ref{fig:diffs} and Table~\ref{table:orgvsbst} at a finer granularity already. 

\begin{figure}[t]
\includegraphics[width=0.4\textwidth]{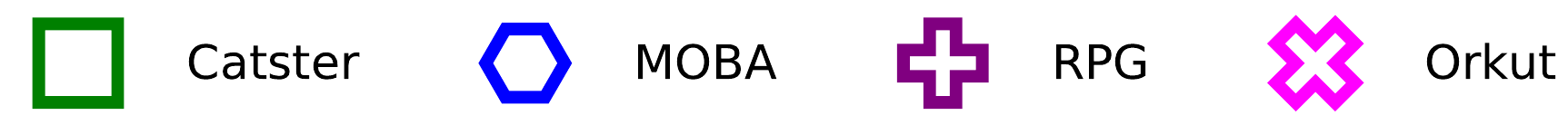}
\vspace{-0.45cm}

\centering
\subfloat[varing $k$]{
\includegraphics[width=0.2\textwidth]{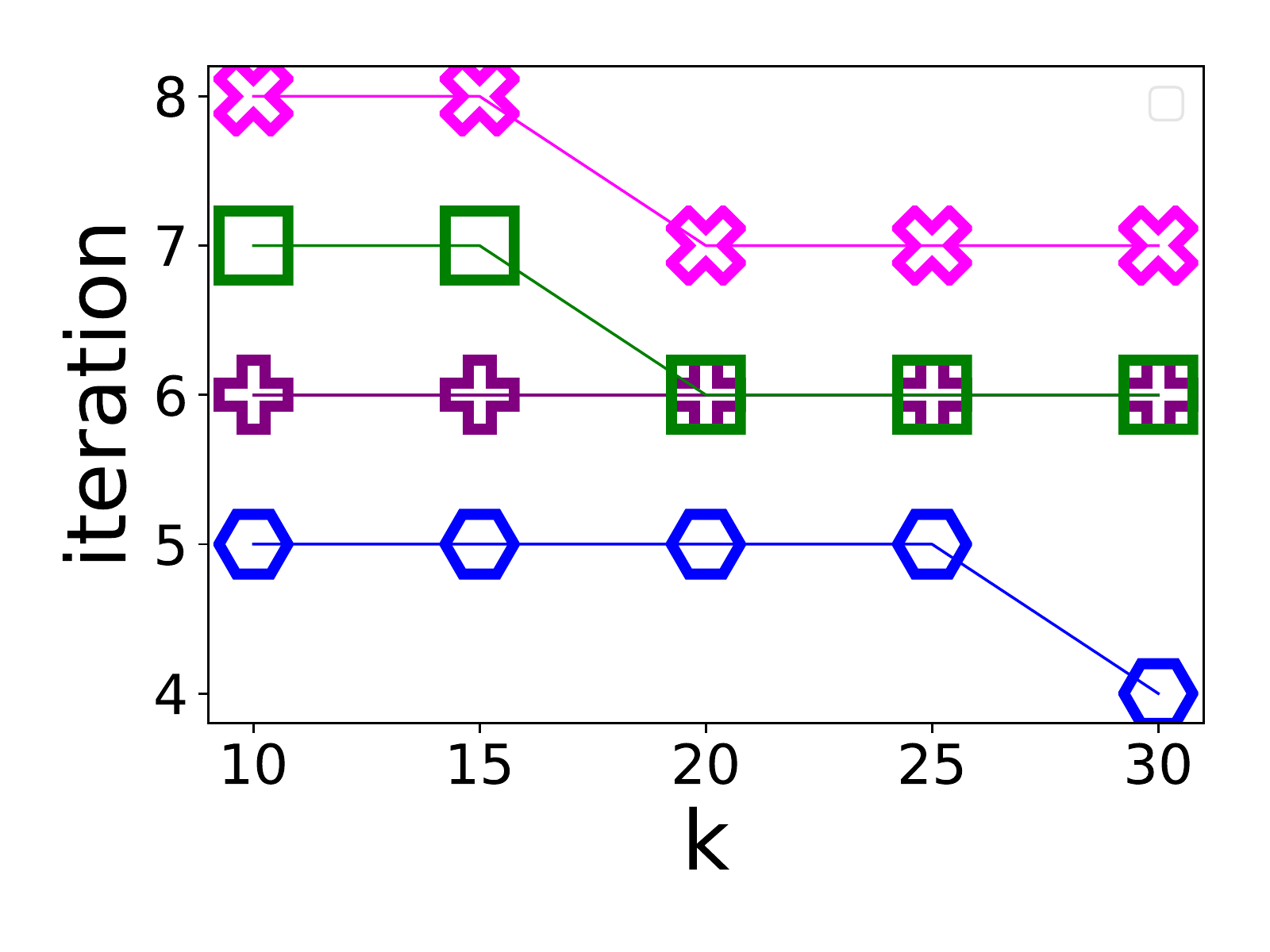}\hspace{-1em}
} 
\hspace{0.5mm}
\subfloat[varing $|S|$]{
\includegraphics[width=0.2\textwidth]{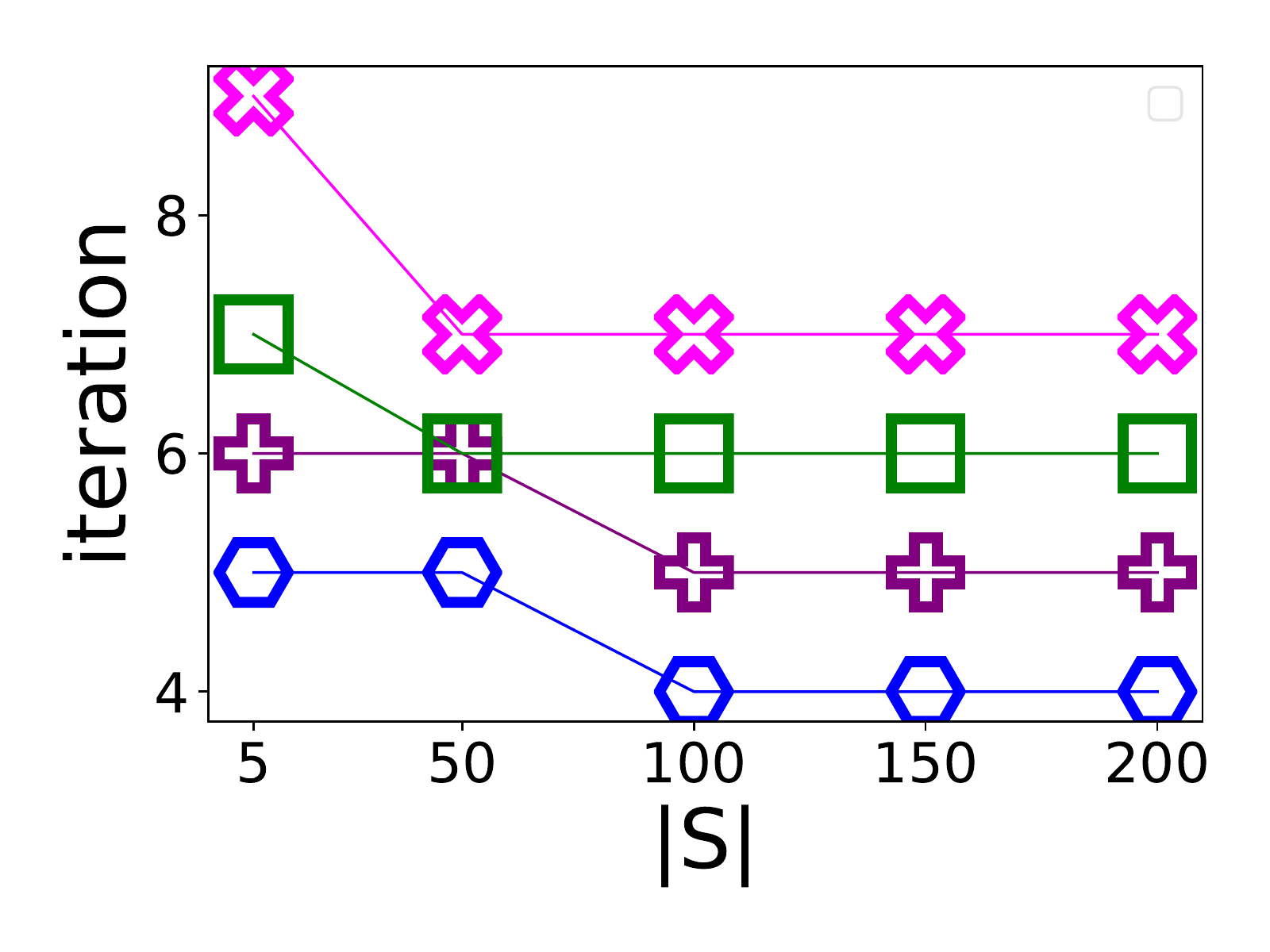}\hspace{-1em}
}

\vspace{-1.em}
\caption{The number of iterations in the expansion stage with different $k$ and $|S|$.}\label{fig:convergence}
\end{figure}

\begin{figure}[t]
\includegraphics[width=0.4\textwidth]{./drawing2/legend2}
\vspace{-0.45cm}

\centering
\subfloat[varing $k$]{
\includegraphics[width=0.21\textwidth]{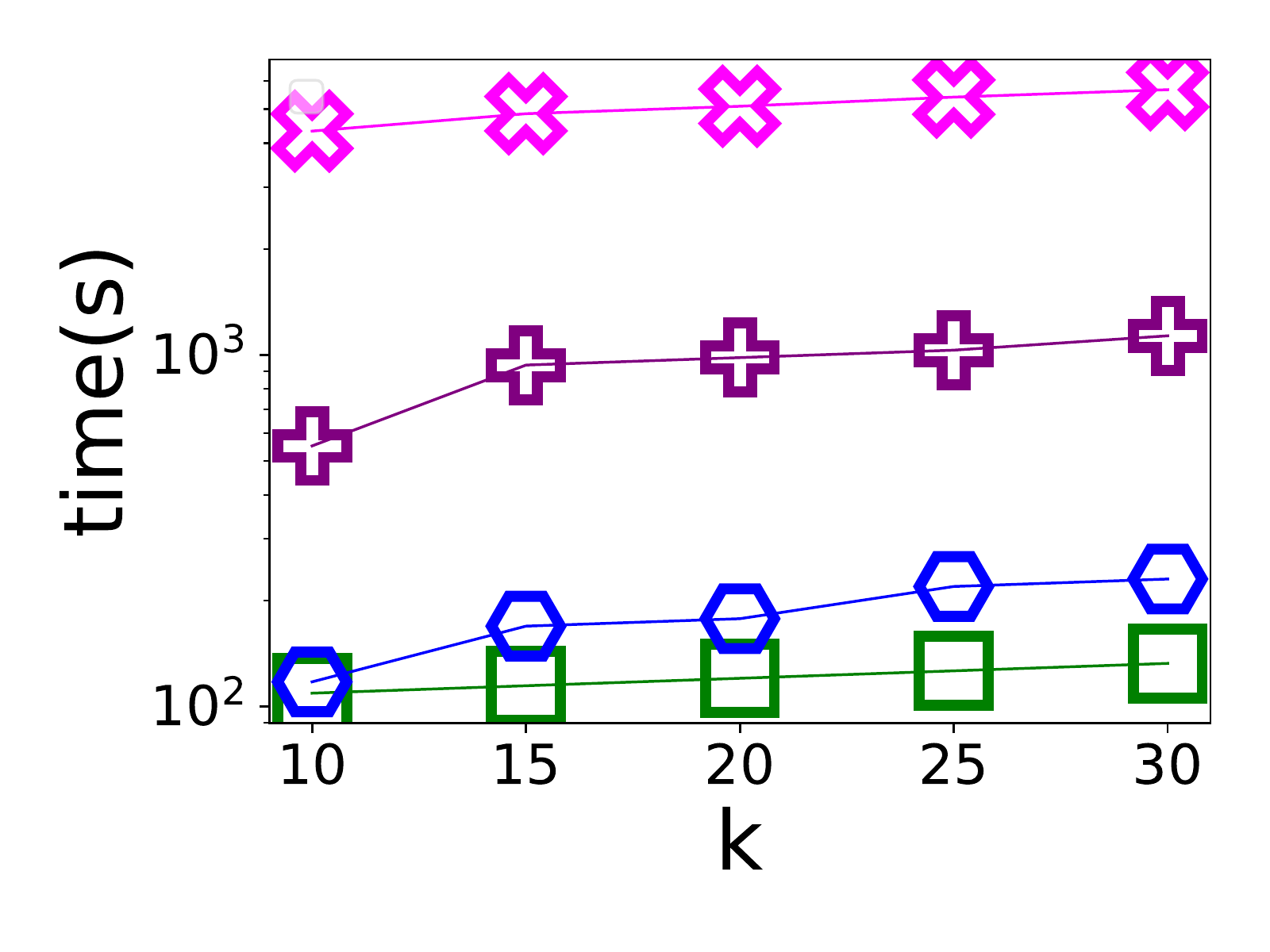}\hspace{-1em}
} 
\hspace{0.5mm}
\subfloat[varing $|S|$]{
\includegraphics[width=0.21\textwidth]{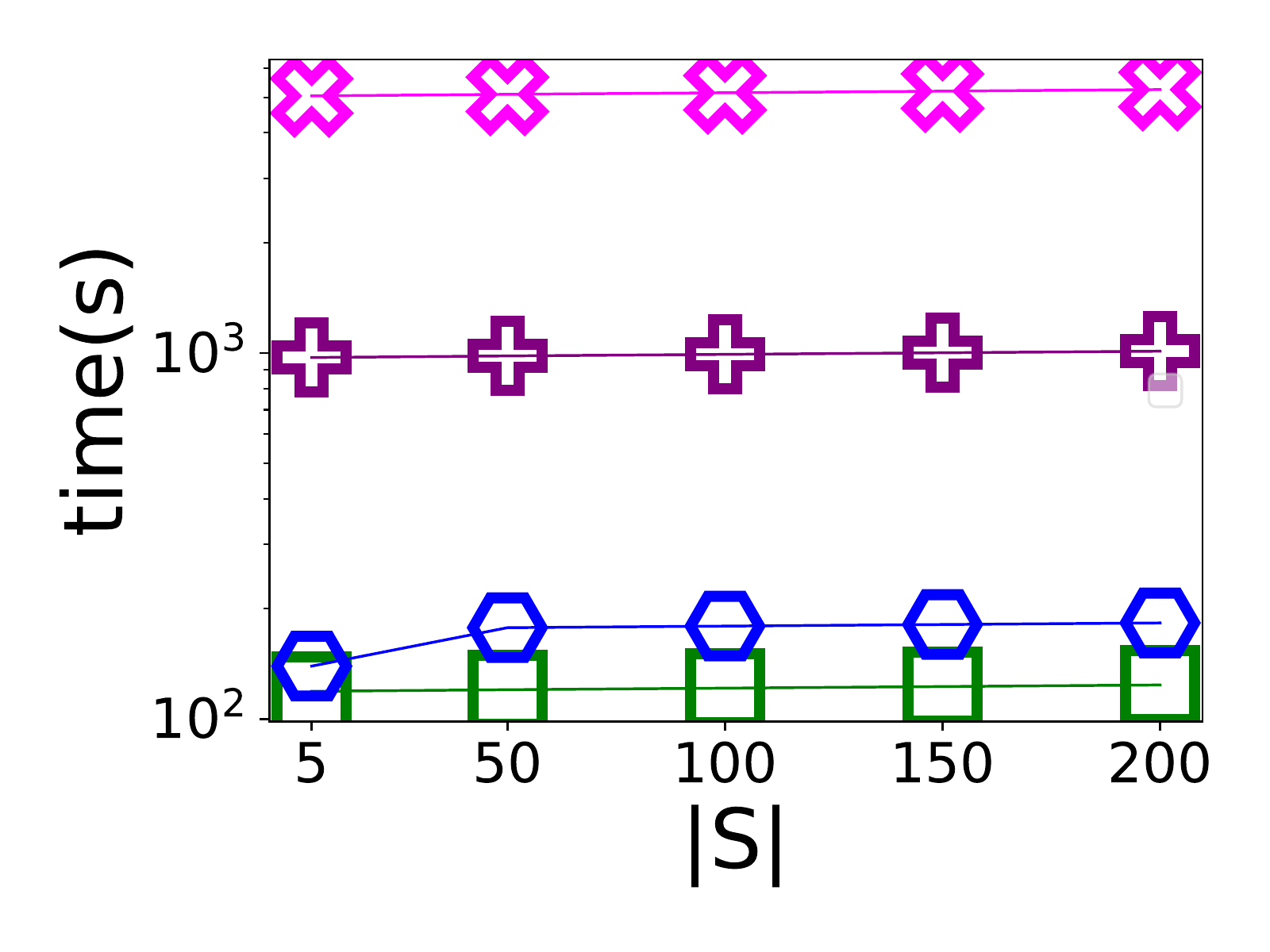}\hspace{-1em}
}

\vspace{-1.em}
\caption{Running time of \outtermethod{} with different $k$ and $|S|$.}\label{fig:efficiency}
\end{figure}

\begin{figure}[t]
\includegraphics[width=0.4\textwidth]{./drawing2/legend2}

\centering
\includegraphics[width=0.2\textwidth]{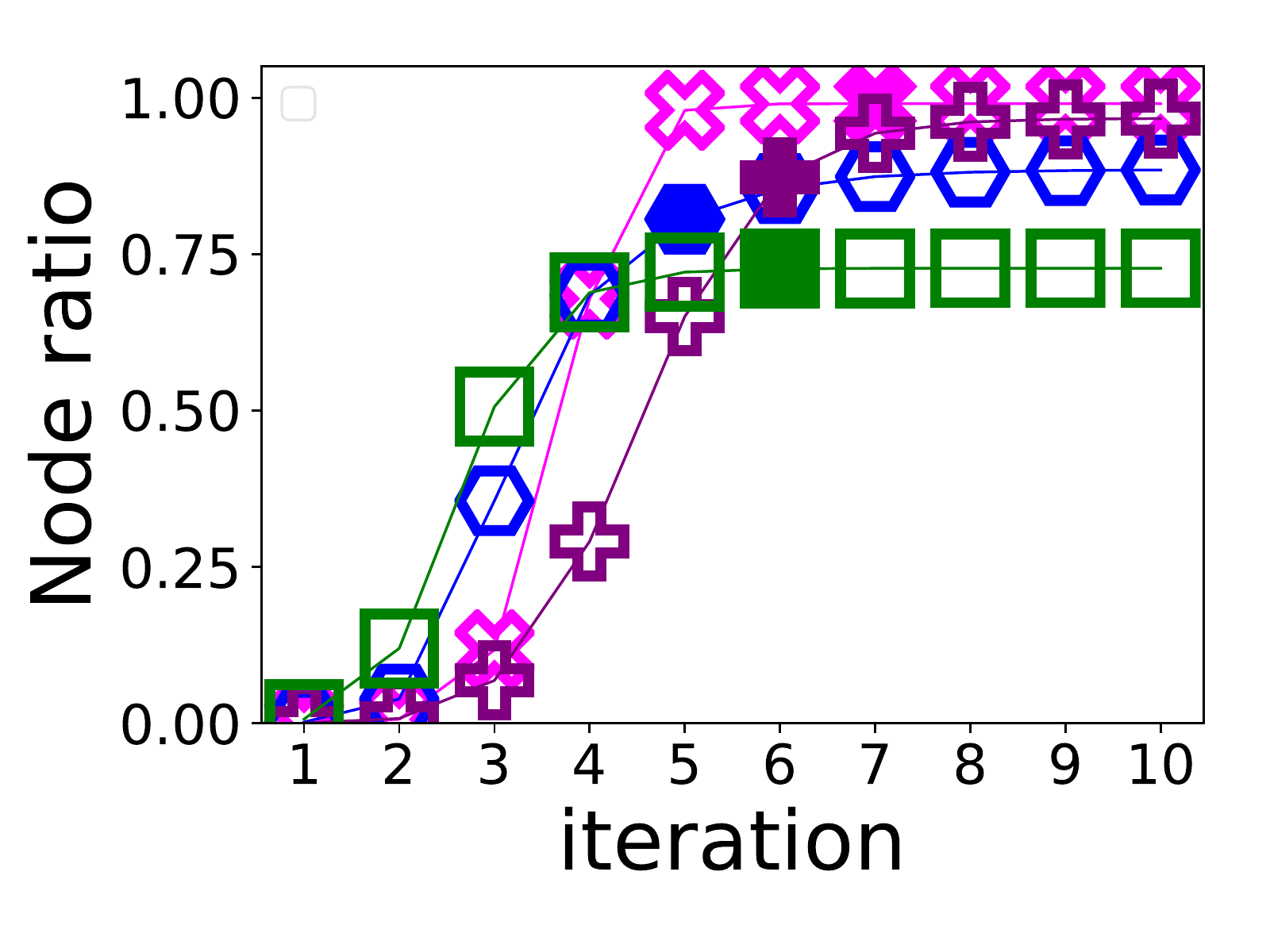}\hspace{-1em}

\vspace{-1.5em}
\caption{The percentage of the nodes in the diffusion network over the one in the original network with different iterations in the expansion stage of \outtermethod{}, where the iterations with solid fill refer to the convergence points under different datasets repsectively.}\label{fig:noderatio}
\end{figure}

\begin{figure}[t]
\includegraphics[width=0.33\textwidth]{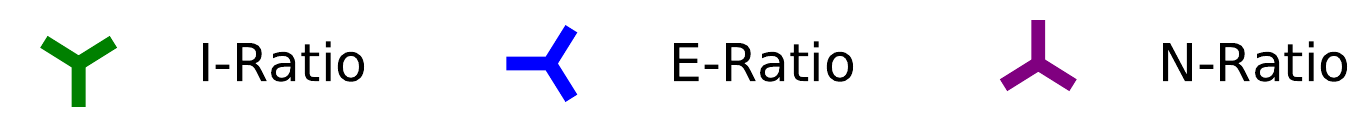}
\vspace{-0.45cm}

\centering
\subfloat[RPG]{
\includegraphics[width=0.2\textwidth]{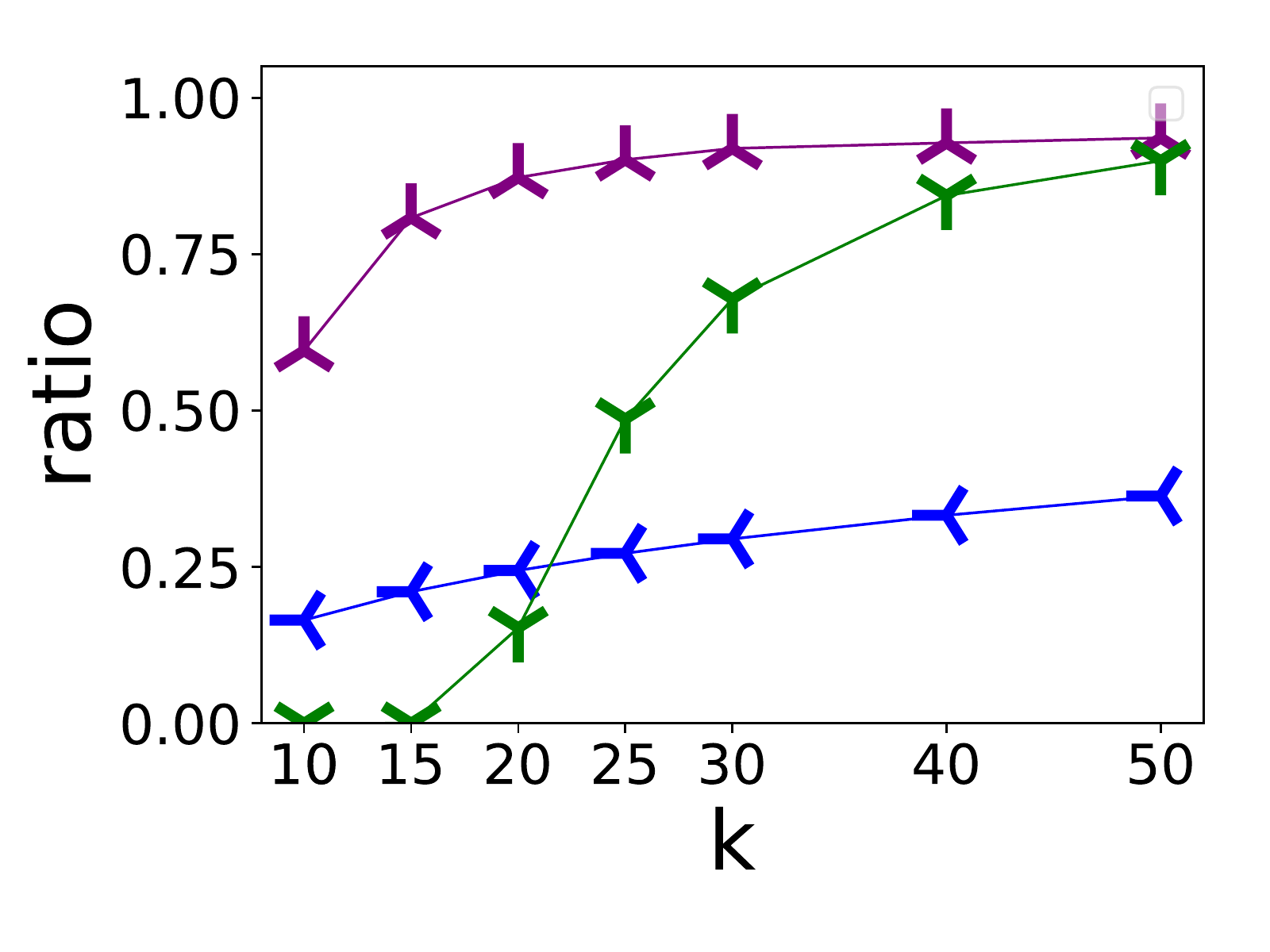}\hspace{-1em}
} 
\hspace{0.5mm}
\subfloat[Orkut]{
\includegraphics[width=0.2\textwidth]{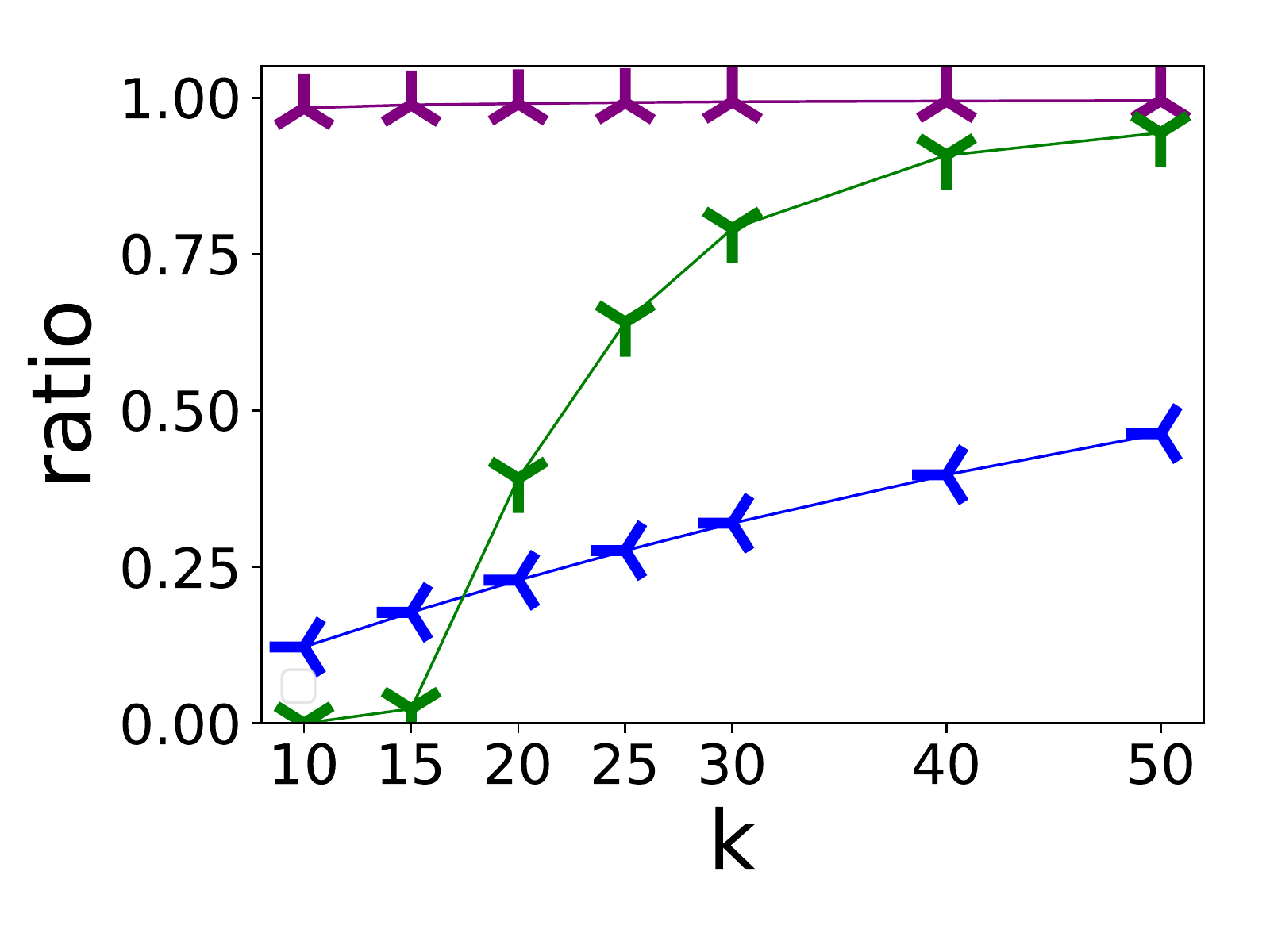}\hspace{-1em}
}

\caption{The influence ratio (I-Ratio), edge ratio (E-Ratio) and node ratio (N-Ratio) in the diffusion network produced by \outtermethod{} with different $k$.}\label{fig:INERatio} 

\end{figure}

\smallskip
\noindent \textit{\textbf{Exp4 - Impact of $k$ and $|S|$ on \outtermethod{} convergence}}. 
Figures~\ref{fig:convergence} (a) and (b) show the number of iterations in the expansion stage of \outtermethod{} when varing $k$ and $|S|$ respectively. Interestingly, when $k$ or $|S|$ increases, the number of iterations for convergence is non-increasing. This is because \outtermethod{} will include more nodes into the current \ksubnetwork{} in each iteration such that it may take fewer iterations to identify all important nodes for influence spread.

\smallskip
\noindent \textit{\textbf{Exp5 - Scalability study on \outtermethod{}}}. 
Figures~\ref{fig:efficiency} (a) and (b) show the runtime of our (full-stage) \outtermethod{} when varing $k$ and $|S|$. As $k$ grows, more nodes tend to be included in the diffusion network and more edges tend to be `filled up' in the filling stage, hence costing more time. On the other hand, increasing $|S|$ does not notably incur more computation costs, for two reasons: 1) we transform the case of multiple seeds into the one of a single virtual node by introducing only $|S|$ extra edges into the diffusion network; 2) the size of the diffusion networks generated under instances of different $|S|$ can be similar as discussed earlier. It is obvious to see that \outtermethod{} is highly scalable and efficient, since it can produce good recommendation on datasets with millions of nodes and hundred millions of edges within two hours (e.g., Orkut). This performance is promising since such recommendation is usually deployed for a long-term usage.

\smallskip
\noindent \textit{\textbf{Exp6 - Impact of iterations on the diffusion network node size}}. 
Figure~\ref{fig:noderatio} shows how the node size of the diffusion network in the expansion stage of \outtermethod{} grows iteratively. Here, the iterations with solid fill refer to the convergence points. Note that the expansion stage is terminated after convergence for other experiments and here we force it to run 10 iterations. It is obvious that, the node size in the diffusion network converges as the influence spread, which indicates the effectiveness of the convergence criteria.

\smallskip
\noindent \textit{\textbf{Exp7 - Case study with large $k$}}. 
Here, we conduct a case study on the two largest datasets, RPG and Orkut. Our goal is to show how the statistics of the diffusion network produced by \outtermethod{} and influence spread of seed nodes in this network will change when we further increase $k$ up to 50. Specifically, we use the \emph{influence ratio} (I-Ratio) to denote the percentage of the influence of seeds in the diffusion network over their full influence in the original network, \emph{edge ratio} (E-Ratio) and \emph{node ratio} (N-Ratio) to denote the percentage of the edge and node size of the diffusion network over the ones of the original network, respectively. As shown in Figure~\ref{fig:INERatio}, as $k$ increases, \outtermethod{} is able to achieve an I-Ratio of 90\% with only an E-Ratio of 36\%, and an I-Ratio of 94\% with only an E-Ratio of 46\% on RPG and Orkut respectively. These again demonstrate the superiority of \outtermethod{}. 
Furthermore, when the N-Ratio converges, the I-Ratio can still notably increase as $k$ grows, and even the converged N-Ratio can be notably smaller than 1 (e.g., on RPG). These indicate that \outtermethod{} can identify important nodes for spreading the influence.

\subsection{Deployment}\label{sec:deploy}
In order to demonstrate the practical effectiveness of our \short{} problem, we deployed our solution into an activity in a multi-player online battle arena game of Tencent of which the friendship network is MOBAX in Table~\ref{table:dataset}.

\smallskip
\noindent \textbf{Summary of the activity}.
Users join this activity by interacting (e.g., sending gifts and gaming invitation) with their friends recommended by the system and will obtain rewards via such interactions. This form of interaction could trigger domino effects such that the interactions can stimulate the message receiver to further interact with other users. The aim of this activity is to improve user retention and interactions. Note that, in this activity, we do not specifically designate seeds. Instead, every user could be a potential and spontaneous seed. Specifically, users who log in the game will see this activity and may perform three possible behaviors: (1) proactively start interactions as a seed, (2) start interactions as a message receiver (i.e., a non-seed) and (3) ignore this activity. 

\smallskip

\noindent \textbf{Edge weights and the diffusion model}. We adopt the same settings as we use for generating edge weights on MOBA and RPG, and generate the diffusion network based on the IC model. Note that the IC model is only used for generating the diffusion network but not used for evaluation since it is very likely that the real-world influence will not spread in a way specified by the IC model or any other classical diffusion models (e.g., Linear Threshold model~\cite{Kempe03}).

\smallskip
\noindent \textbf{Methods for comparison}. Four methods, namely \method{Random}, \method{Degree}, \method{FoF} and \outtermethod{}, are deployed. For evaluation, we use an online A/B testing. Specifically, it randomly assigns each online user to one of these four methods which recommend a list of friends for each assigned user.
The list size $k$ is set as 20 for all methods because, in this application, only 20 friends are shown due to several reasons (e.g., the screen size and resolution settings of most mobile devices).
Based on our observations made in Section~\ref{sec:exp-closed}, the k-subnetwork generated based on a given seed set could also be effective for different seed sets.
Thus, for \outtermethod{}, we first generate the \ksubnetwork{} with top 10\% of assigned users with the highest degrees as the seeds, and then use this generated \ksubnetwork{} for this activity where every user could be a potential seed or non-seed.

\smallskip
\noindent \textbf{Experimental results}. The number of users who interact with the recommended friends directly indicates the quality of the underlying recommendation system and can effectively reflect the potential size of users who are impacted by others to perform interactions. 
Thus, to evaluate the performance, we adopt the Click-through Rate (CTR) which is a very popular evaluation metric for recommendation systems in industry. The CTR denotes the number of users who interacted with recommended friends divided by the total number of users who log in the game during this activity. 
The CTR achieved by \method{Random}, \method{Degree}, \method{FoF} and \outtermethod{} are 81.87\%, 82.73\%, 82.81\% and \textbf{88.14\%} respectively, which again demonstrates the effectiveness of \outtermethod{}. Despite the effectiveness of CTR, there still exist some other evaluation metrics (e.g., the actual influence spread of seeds via knowing the propagation paths and who are the seeds) worthy of explorations. However, due to the privacy restriction of this activity, we have no access to the detailed information in user logs to compute the performance under other potential metrics. In future, we would like to see how our proposed method will perform in other activities under various evaluation metrics.

%% file: tex/appendix.tex
\appendix

%% file: draft.bbl

\begin{thebibliography}{63}


\ifx \showCODEN    \undefined \def \showCODEN     #1{\unskip}     \fi
\ifx \showDOI      \undefined \def \showDOI       #1{#1}\fi
\ifx \showISBNx    \undefined \def \showISBNx     #1{\unskip}     \fi
\ifx \showISBNxiii \undefined \def \showISBNxiii  #1{\unskip}     \fi
\ifx \showISSN     \undefined \def \showISSN      #1{\unskip}     \fi
\ifx \showLCCN     \undefined \def \showLCCN      #1{\unskip}     \fi
\ifx \shownote     \undefined \def \shownote      #1{#1}          \fi
\ifx \showarticletitle \undefined \def \showarticletitle #1{#1}   \fi
\ifx \showURL      \undefined \def \showURL       {\relax}        \fi
\providecommand\bibfield[2]{#2}
\providecommand\bibinfo[2]{#2}
\providecommand\natexlab[1]{#1}
\providecommand\showeprint[2][]{arXiv:#2}

\bibitem[\protect\citeauthoryear{??}{blo}{2014}]%
        {block2}
 \bibinfo{year}{2014}\natexlab{}.
\newblock
  \bibinfo{howpublished}{\url{https://business.sohu.com/20140624/n401244299.shtml}}.
\newblock


\bibitem[\protect\citeauthoryear{??}{fom}{2018}]%
        {fomo}
 \bibinfo{year}{2018}\natexlab{}.
\newblock
  \bibinfo{howpublished}{\url{https://www.postbeyond.com/blog/millennials-genz-social-media/}}.
\newblock


\bibitem[\protect\citeauthoryear{??}{tre}{2018}]%
        {trend}
 \bibinfo{year}{2018}\natexlab{}.
\newblock
  \bibinfo{howpublished}{\url{https://medium.com/@lorenabarquin/are-closed-social-media-platforms-the-future-of-social-3a5b0cbea025}}.
\newblock


\bibitem[\protect\citeauthoryear{??}{sur}{2018}]%
        {survery}
 \bibinfo{year}{2018}\natexlab{}.
\newblock
  \bibinfo{howpublished}{\url{https://www.warc.com/newsandopinion/news/the_new_facebooks_the_trend_towards_a_closed_social_media/40929}}.
\newblock


\bibitem[\protect\citeauthoryear{??}{exp}{2018}]%
        {expose}
 \bibinfo{year}{2018}\natexlab{}.
\newblock
  \bibinfo{howpublished}{\url{https://www.quora.com/Why-are-some-people-not-interested-in-exposing-themselves-on-social-media}}.
\newblock


\bibitem[\protect\citeauthoryear{??}{tow}{2021}]%
        {towardclosed3}
 \bibinfo{year}{2021}\natexlab{}.
\newblock
  \bibinfo{howpublished}{\url{https://www.tailwindapp.com/blog/private-on-pinterest}}.
\newblock


\bibitem[\protect\citeauthoryear{??}{blo}{2021}]%
        {block}
 \bibinfo{year}{2021}\natexlab{}.
\newblock \bibinfo{howpublished}{\url{https://zhuanlan.zhihu.com/p/82896779}}.
\newblock


\bibitem[\protect\citeauthoryear{??}{tow}{2022a}]%
        {towardclosed1}
 \bibinfo{year}{2022}\natexlab{a}.
\newblock
  \bibinfo{howpublished}{\url{https://www.facebook.com/help/233739099984085}}.
\newblock


\bibitem[\protect\citeauthoryear{??}{tow}{2022b}]%
        {towardclosed2}
 \bibinfo{year}{2022}\natexlab{b}.
\newblock
  \bibinfo{howpublished}{\url{https://help.twitter.com/en/safety-and-security/public-and-protected-tweets}}.
\newblock


\bibitem[\protect\citeauthoryear{??}{nam}{2022}]%
        {naming1}
 \bibinfo{year}{2022}\natexlab{}.
\newblock
  \bibinfo{howpublished}{\url{https://techalignment.com/closed-versus-open-social-networks/}}.
\newblock


\bibitem[\protect\citeauthoryear{??}{TR}{2022}]%
        {TR}
 \bibinfo{year}{2022}\natexlab{}.
\newblock \bibinfo{howpublished}{\url{https://github.com/rmitbggroup/IMCSN}}.
\newblock


\bibitem[\protect\citeauthoryear{Anderson and May}{Anderson and May}{1992}]%
        {anderson1992infectious}
\bibfield{author}{\bibinfo{person}{Roy~M Anderson} {and}
  \bibinfo{person}{Robert~M May}.} \bibinfo{year}{1992}\natexlab{}.
\newblock \bibinfo{booktitle}{\emph{Infectious diseases of humans: dynamics and
  control}}.
\newblock


\bibitem[\protect\citeauthoryear{Aslay, Lakshmanan, Lu, and Xiao}{Aslay
  et~al\mbox{.}}{2018}]%
        {aslay2018influence}
\bibfield{author}{\bibinfo{person}{Cigdem Aslay}, \bibinfo{person}{Laks~VS
  Lakshmanan}, \bibinfo{person}{Wei Lu}, {and} \bibinfo{person}{Xiaokui Xiao}.}
  \bibinfo{year}{2018}\natexlab{}.
\newblock \showarticletitle{Influence maximization in online social networks}.
  In \bibinfo{booktitle}{\emph{WSDM}}. \bibinfo{pages}{775--776}.
\newblock


\bibitem[\protect\citeauthoryear{Banerjee, Jenamani, and Pratihar}{Banerjee
  et~al\mbox{.}}{2019}]%
        {banerjee2019combim}
\bibfield{author}{\bibinfo{person}{Suman Banerjee}, \bibinfo{person}{Mamata
  Jenamani}, {and} \bibinfo{person}{Dilip~Kumar Pratihar}.}
  \bibinfo{year}{2019}\natexlab{}.
\newblock \showarticletitle{ComBIM: A community-based solution approach for the
  Budgeted Influence Maximization Problem}.
\newblock \bibinfo{journal}{\emph{Expert Systems with Applications}}
  \bibinfo{volume}{125} (\bibinfo{year}{2019}), \bibinfo{pages}{1--13}.
\newblock


\bibitem[\protect\citeauthoryear{Bevilacqua and Lakshmanan}{Bevilacqua and
  Lakshmanan}{2021}]%
        {bevilacqua2021fractional}
\bibfield{author}{\bibinfo{person}{Glenn~S Bevilacqua} {and}
  \bibinfo{person}{Laks~VS Lakshmanan}.} \bibinfo{year}{2021}\natexlab{}.
\newblock \showarticletitle{A fractional memory-efficient approach for online
  continuous-time influence maximization}.
\newblock \bibinfo{journal}{\emph{The VLDB Journal}} (\bibinfo{year}{2021}),
  \bibinfo{pages}{1--27}.
\newblock


\bibitem[\protect\citeauthoryear{Bian, Guo, Wang, and Yu}{Bian
  et~al\mbox{.}}{2020}]%
        {bian2020efficient}
\bibfield{author}{\bibinfo{person}{Song Bian}, \bibinfo{person}{Qintian Guo},
  \bibinfo{person}{Sibo Wang}, {and} \bibinfo{person}{Jeffrey~Xu Yu}.}
  \bibinfo{year}{2020}\natexlab{}.
\newblock \showarticletitle{Efficient algorithms for budgeted influence
  maximization on massive social networks}.
\newblock \bibinfo{journal}{\emph{VLDB}} \bibinfo{volume}{13},
  \bibinfo{number}{9} (\bibinfo{year}{2020}), \bibinfo{pages}{1498--1510}.
\newblock


\bibitem[\protect\citeauthoryear{Borgs, Brautbar, Chayes, and Lucier}{Borgs
  et~al\mbox{.}}{2014}]%
        {borgs2014maximizing}
\bibfield{author}{\bibinfo{person}{Christian Borgs}, \bibinfo{person}{Michael
  Brautbar}, \bibinfo{person}{Jennifer Chayes}, {and} \bibinfo{person}{Brendan
  Lucier}.} \bibinfo{year}{2014}\natexlab{}.
\newblock \showarticletitle{Maximizing social influence in nearly optimal
  time}. In \bibinfo{booktitle}{\emph{SODA}}. \bibinfo{pages}{946--957}.
\newblock


\bibitem[\protect\citeauthoryear{Cai, Li, Mian, Sellis, Yu, et~al\mbox{.}}{Cai
  et~al\mbox{.}}{2020}]%
        {cai2020target}
\bibfield{author}{\bibinfo{person}{Taotao Cai}, \bibinfo{person}{Jianxin Li},
  \bibinfo{person}{Ajmal~S Mian}, \bibinfo{person}{Timos Sellis},
  \bibinfo{person}{Jeffrey~Xu Yu}, {et~al\mbox{.}}}
  \bibinfo{year}{2020}\natexlab{}.
\newblock \showarticletitle{Target-aware holistic influence maximization in
  spatial social networks}.
\newblock \bibinfo{journal}{\emph{TKDE}} (\bibinfo{year}{2020}).
\newblock


\bibitem[\protect\citeauthoryear{Castellano and Pastor-Satorras}{Castellano and
  Pastor-Satorras}{2010}]%
        {castellano2010thresholds}
\bibfield{author}{\bibinfo{person}{Claudio Castellano} {and}
  \bibinfo{person}{Romualdo Pastor-Satorras}.} \bibinfo{year}{2010}\natexlab{}.
\newblock \showarticletitle{Thresholds for epidemic spreading in networks}.
\newblock \bibinfo{journal}{\emph{Physical review letters}}
  \bibinfo{volume}{105}, \bibinfo{number}{21} (\bibinfo{year}{2010}),
  \bibinfo{pages}{218701}.
\newblock


\bibitem[\protect\citeauthoryear{Cautis, Maniu, and Tziortziotis}{Cautis
  et~al\mbox{.}}{2019}]%
        {cautis2019adaptive}
\bibfield{author}{\bibinfo{person}{Bogdan Cautis}, \bibinfo{person}{Silviu
  Maniu}, {and} \bibinfo{person}{Nikolaos Tziortziotis}.}
  \bibinfo{year}{2019}\natexlab{}.
\newblock \showarticletitle{Adaptive influence maximization}. In
  \bibinfo{booktitle}{\emph{SIGKDD}}. \bibinfo{pages}{3185--3186}.
\newblock


\bibitem[\protect\citeauthoryear{Chaoji, Ranu, Rastogi, and Bhatt}{Chaoji
  et~al\mbox{.}}{2012}]%
        {chaoji2012recommendations}
\bibfield{author}{\bibinfo{person}{Vineet Chaoji}, \bibinfo{person}{Sayan
  Ranu}, \bibinfo{person}{Rajeev Rastogi}, {and} \bibinfo{person}{Rushi
  Bhatt}.} \bibinfo{year}{2012}\natexlab{}.
\newblock \showarticletitle{Recommendations to boost content spread in social
  networks}. In \bibinfo{booktitle}{\emph{WWW}}. \bibinfo{pages}{529--538}.
\newblock


\bibitem[\protect\citeauthoryear{Chen, Wang, and Wang}{Chen
  et~al\mbox{.}}{2010a}]%
        {chen2010scalable}
\bibfield{author}{\bibinfo{person}{Wei Chen}, \bibinfo{person}{Chi Wang}, {and}
  \bibinfo{person}{Yajun Wang}.} \bibinfo{year}{2010}\natexlab{a}.
\newblock \showarticletitle{Scalable influence maximization for prevalent viral
  marketing in large-scale social networks}. In
  \bibinfo{booktitle}{\emph{SIGKDD}}. \bibinfo{pages}{1029--1038}.
\newblock


\bibitem[\protect\citeauthoryear{Chen, Wang, and Yang}{Chen
  et~al\mbox{.}}{2009}]%
        {chen2009efficient}
\bibfield{author}{\bibinfo{person}{Wei Chen}, \bibinfo{person}{Yajun Wang},
  {and} \bibinfo{person}{Siyu Yang}.} \bibinfo{year}{2009}\natexlab{}.
\newblock \showarticletitle{Efficient influence maximization in social
  networks}. In \bibinfo{booktitle}{\emph{SIGKDD}}. \bibinfo{pages}{199--208}.
\newblock


\bibitem[\protect\citeauthoryear{Chen, Yuan, and Zhang}{Chen
  et~al\mbox{.}}{2010b}]%
        {chen2010scalable2}
\bibfield{author}{\bibinfo{person}{Wei Chen}, \bibinfo{person}{Yifei Yuan},
  {and} \bibinfo{person}{Li Zhang}.} \bibinfo{year}{2010}\natexlab{b}.
\newblock \showarticletitle{Scalable influence maximization in social networks
  under the linear threshold model}. In \bibinfo{booktitle}{\emph{ICDM}}.
  \bibinfo{pages}{88--97}.
\newblock


\bibitem[\protect\citeauthoryear{Cheng, Shen, Huang, Zhang, and Cheng}{Cheng
  et~al\mbox{.}}{2013}]%
        {cheng2013staticgreedy}
\bibfield{author}{\bibinfo{person}{Suqi Cheng}, \bibinfo{person}{Huawei Shen},
  \bibinfo{person}{Junming Huang}, \bibinfo{person}{Guoqing Zhang}, {and}
  \bibinfo{person}{Xueqi Cheng}.} \bibinfo{year}{2013}\natexlab{}.
\newblock \showarticletitle{Staticgreedy: solving the scalability-accuracy
  dilemma in influence maximization}. In \bibinfo{booktitle}{\emph{CIKM}}.
  \bibinfo{pages}{509--518}.
\newblock


\bibitem[\protect\citeauthoryear{Choi and Lee}{Choi and Lee}{2017}]%
        {naming2}
\bibfield{author}{\bibinfo{person}{Boreum Choi} {and} \bibinfo{person}{Inseong
  Lee}.} \bibinfo{year}{2017}\natexlab{}.
\newblock \showarticletitle{Trust in open versus closed social media: The
  relative influence of user-and marketer-generated content in social network
  services on customer trust}.
\newblock \bibinfo{journal}{\emph{Telematics and Informatics}}
  \bibinfo{volume}{34}, \bibinfo{number}{5} (\bibinfo{year}{2017}),
  \bibinfo{pages}{550--559}.
\newblock


\bibitem[\protect\citeauthoryear{Cohen, Erez, Havlinl, Newman, Barab{\'a}si,
  Watts, et~al\mbox{.}}{Cohen et~al\mbox{.}}{2011}]%
        {cohen2011resilience}
\bibfield{author}{\bibinfo{person}{Reuven Cohen}, \bibinfo{person}{Keren Erez},
  \bibinfo{person}{Shlomo Havlinl}, \bibinfo{person}{Mark Newman},
  \bibinfo{person}{Albert-L{\'a}szl{\'o} Barab{\'a}si},
  \bibinfo{person}{Duncan~J Watts}, {et~al\mbox{.}}}
  \bibinfo{year}{2011}\natexlab{}.
\newblock \showarticletitle{Resilience of the internet to random breakdowns}.
\newblock In \bibinfo{booktitle}{\emph{The Structure and Dynamics of
  Networks}}. \bibinfo{pages}{507--509}.
\newblock


\bibitem[\protect\citeauthoryear{Collection}{Collection}{2017}]%
        {konect}
\bibfield{author}{\bibinfo{person}{The Koblenz~Network Collection}.}
  \bibinfo{year}{2017}\natexlab{}.
\newblock \bibinfo{howpublished}{\url{http://konect.uni-koblenz.de}}.
\newblock


\bibitem[\protect\citeauthoryear{Cor{\'o}, D’angelo, and Velaj}{Cor{\'o}
  et~al\mbox{.}}{2021}]%
        {coro2021link}
\bibfield{author}{\bibinfo{person}{Federico Cor{\'o}},
  \bibinfo{person}{Gianlorenzo D’angelo}, {and} \bibinfo{person}{Yllka
  Velaj}.} \bibinfo{year}{2021}\natexlab{}.
\newblock \showarticletitle{Link Recommendation for Social Influence
  Maximization}.
\newblock \bibinfo{journal}{\emph{TKDD}} \bibinfo{volume}{15},
  \bibinfo{number}{6} (\bibinfo{year}{2021}), \bibinfo{pages}{1--23}.
\newblock


\bibitem[\protect\citeauthoryear{D'Angelo, Severini, and Velaj}{D'Angelo
  et~al\mbox{.}}{2019}]%
        {d2019recommending}
\bibfield{author}{\bibinfo{person}{Gianlorenzo D'Angelo},
  \bibinfo{person}{Lorenzo Severini}, {and} \bibinfo{person}{Yllka Velaj}.}
  \bibinfo{year}{2019}\natexlab{}.
\newblock \showarticletitle{Recommending links through influence maximization}.
\newblock \bibinfo{journal}{\emph{Theor. Comput. Sci.}}  \bibinfo{volume}{764}
  (\bibinfo{year}{2019}), \bibinfo{pages}{30--41}.
\newblock


\bibitem[\protect\citeauthoryear{Galhotra, Arora, and Roy}{Galhotra
  et~al\mbox{.}}{2016}]%
        {galhotra2016holistic}
\bibfield{author}{\bibinfo{person}{Sainyam Galhotra}, \bibinfo{person}{Akhil
  Arora}, {and} \bibinfo{person}{Shourya Roy}.}
  \bibinfo{year}{2016}\natexlab{}.
\newblock \showarticletitle{Holistic influence maximization: Combining
  scalability and efficiency with opinion-aware models}. In
  \bibinfo{booktitle}{\emph{SIGMOD}}. \bibinfo{pages}{1077--1088}.
\newblock


\bibitem[\protect\citeauthoryear{Goldenberg, Libai, and Muller}{Goldenberg
  et~al\mbox{.}}{2001}]%
        {goldenberg2001using}
\bibfield{author}{\bibinfo{person}{Jacob Goldenberg}, \bibinfo{person}{Barak
  Libai}, {and} \bibinfo{person}{Eitan Muller}.}
  \bibinfo{year}{2001}\natexlab{}.
\newblock \showarticletitle{Using complex systems analysis to advance marketing
  theory development: Modeling heterogeneity effects on new product growth
  through stochastic cellular automata}.
\newblock \bibinfo{journal}{\emph{Academy of Marketing Science Review}}
  \bibinfo{volume}{9}, \bibinfo{number}{3} (\bibinfo{year}{2001}),
  \bibinfo{pages}{1--18}.
\newblock


\bibitem[\protect\citeauthoryear{Goundan and Schulz}{Goundan and
  Schulz}{2007}]%
        {goundan2007revisiting}
\bibfield{author}{\bibinfo{person}{Pranava~R Goundan} {and}
  \bibinfo{person}{Andreas~S Schulz}.} \bibinfo{year}{2007}\natexlab{}.
\newblock \showarticletitle{Revisiting the greedy approach to submodular set
  function maximization}.
\newblock \bibinfo{journal}{\emph{Optimization online}} (\bibinfo{year}{2007}),
  \bibinfo{pages}{1--25}.
\newblock


\bibitem[\protect\citeauthoryear{Goyal, Lu, and Lakshmanan}{Goyal
  et~al\mbox{.}}{2011a}]%
        {goyal2011celf++}
\bibfield{author}{\bibinfo{person}{Amit Goyal}, \bibinfo{person}{Wei Lu}, {and}
  \bibinfo{person}{Laks~VS Lakshmanan}.} \bibinfo{year}{2011}\natexlab{a}.
\newblock \showarticletitle{Celf++: optimizing the greedy algorithm for
  influence maximization in social networks}. In
  \bibinfo{booktitle}{\emph{WWW}}. \bibinfo{pages}{47--48}.
\newblock


\bibitem[\protect\citeauthoryear{Goyal, Lu, and Lakshmanan}{Goyal
  et~al\mbox{.}}{2011b}]%
        {goyal2011simpath}
\bibfield{author}{\bibinfo{person}{Amit Goyal}, \bibinfo{person}{Wei Lu}, {and}
  \bibinfo{person}{Laks~VS Lakshmanan}.} \bibinfo{year}{2011}\natexlab{b}.
\newblock \showarticletitle{Simpath: An efficient algorithm for influence
  maximization under the linear threshold model}. In
  \bibinfo{booktitle}{\emph{ICDM}}. \bibinfo{pages}{211--220}.
\newblock


\bibitem[\protect\citeauthoryear{Granovetter}{Granovetter}{1978}]%
        {granovetter1978threshold}
\bibfield{author}{\bibinfo{person}{Mark Granovetter}.}
  \bibinfo{year}{1978}\natexlab{}.
\newblock \showarticletitle{Threshold models of collective behavior}.
\newblock \bibinfo{journal}{\emph{American journal of sociology}}
  \bibinfo{volume}{83}, \bibinfo{number}{6} (\bibinfo{year}{1978}),
  \bibinfo{pages}{1420--1443}.
\newblock


\bibitem[\protect\citeauthoryear{Han, Huang, Xiao, Tang, Sun, and Tang}{Han
  et~al\mbox{.}}{2018}]%
        {han2018efficient}
\bibfield{author}{\bibinfo{person}{Kai Han}, \bibinfo{person}{Keke Huang},
  \bibinfo{person}{Xiaokui Xiao}, \bibinfo{person}{Jing Tang},
  \bibinfo{person}{Aixin Sun}, {and} \bibinfo{person}{Xueyan Tang}.}
  \bibinfo{year}{2018}\natexlab{}.
\newblock \showarticletitle{Efficient algorithms for adaptive influence
  maximization}.
\newblock \bibinfo{journal}{\emph{VLDB}} \bibinfo{volume}{11},
  \bibinfo{number}{9} (\bibinfo{year}{2018}), \bibinfo{pages}{1029--1040}.
\newblock


\bibitem[\protect\citeauthoryear{He, Wang, Lei, Huang, Cai, and Ma}{He
  et~al\mbox{.}}{2019}]%
        {he2019tifim}
\bibfield{author}{\bibinfo{person}{Qiang He}, \bibinfo{person}{Xingwei Wang},
  \bibinfo{person}{Zhencheng Lei}, \bibinfo{person}{Min Huang},
  \bibinfo{person}{Yuliang Cai}, {and} \bibinfo{person}{Lianbo Ma}.}
  \bibinfo{year}{2019}\natexlab{}.
\newblock \showarticletitle{TIFIM: A two-stage iterative framework for
  influence maximization in social networks}.
\newblock \bibinfo{journal}{\emph{Appl. Math. Comput.}}  \bibinfo{volume}{354}
  (\bibinfo{year}{2019}), \bibinfo{pages}{338--352}.
\newblock


\bibitem[\protect\citeauthoryear{Huang, Shen, Meng, Chang, and He}{Huang
  et~al\mbox{.}}{2019b}]%
        {huang2019community}
\bibfield{author}{\bibinfo{person}{Huimin Huang}, \bibinfo{person}{Hong Shen},
  \bibinfo{person}{Zaiqiao Meng}, \bibinfo{person}{Huajian Chang}, {and}
  \bibinfo{person}{Huaiwen He}.} \bibinfo{year}{2019}\natexlab{b}.
\newblock \showarticletitle{Community-based influence maximization for viral
  marketing}.
\newblock \bibinfo{journal}{\emph{Applied Intelligence}} \bibinfo{volume}{49},
  \bibinfo{number}{6} (\bibinfo{year}{2019}), \bibinfo{pages}{2137--2150}.
\newblock


\bibitem[\protect\citeauthoryear{Huang, Tang, Han, Xiao, Chen, Sun, Tang, and
  Lim}{Huang et~al\mbox{.}}{2020}]%
        {huang2020efficient}
\bibfield{author}{\bibinfo{person}{Keke Huang}, \bibinfo{person}{Jing Tang},
  \bibinfo{person}{Kai Han}, \bibinfo{person}{Xiaokui Xiao},
  \bibinfo{person}{Wei Chen}, \bibinfo{person}{Aixin Sun},
  \bibinfo{person}{Xueyan Tang}, {and} \bibinfo{person}{Andrew Lim}.}
  \bibinfo{year}{2020}\natexlab{}.
\newblock \showarticletitle{Efficient approximation algorithms for adaptive
  influence maximization}.
\newblock \bibinfo{journal}{\emph{The VLDB Journal}} \bibinfo{volume}{29},
  \bibinfo{number}{6} (\bibinfo{year}{2020}), \bibinfo{pages}{1385--1406}.
\newblock


\bibitem[\protect\citeauthoryear{Huang}{Huang}{2021}]%
        {huang2021capturing}
\bibfield{author}{\bibinfo{person}{Shixun Huang}.}
  \bibinfo{year}{2021}\natexlab{}.
\newblock \emph{\bibinfo{title}{Capturing and leveraging collective behavior
  for large-scale social networks analysis}}.
\newblock \bibinfo{thesistype}{Ph.D. Dissertation}. \bibinfo{school}{RMIT
  University}.
\newblock


\bibitem[\protect\citeauthoryear{Huang, Bao, Culpepper, and Zhang}{Huang
  et~al\mbox{.}}{2019a}]%
        {huang2019finding}
\bibfield{author}{\bibinfo{person}{Shixun Huang}, \bibinfo{person}{Zhifeng
  Bao}, \bibinfo{person}{J~Shane Culpepper}, {and} \bibinfo{person}{Bang
  Zhang}.} \bibinfo{year}{2019}\natexlab{a}.
\newblock \showarticletitle{Finding temporal influential users over evolving
  social networks}. In \bibinfo{booktitle}{\emph{2019 IEEE 35th international
  conference on data engineering (ICDE)}}. IEEE, \bibinfo{pages}{398--409}.
\newblock


\bibitem[\protect\citeauthoryear{Jung, Heo, and Chen}{Jung
  et~al\mbox{.}}{2012}]%
        {jung2012irie}
\bibfield{author}{\bibinfo{person}{Kyomin Jung}, \bibinfo{person}{Wooram Heo},
  {and} \bibinfo{person}{Wei Chen}.} \bibinfo{year}{2012}\natexlab{}.
\newblock \showarticletitle{Irie: Scalable and robust influence maximization in
  social networks}. In \bibinfo{booktitle}{\emph{ICDM}}.
  \bibinfo{pages}{918--923}.
\newblock


\bibitem[\protect\citeauthoryear{Kempe, Kleinberg, and Tardos}{Kempe
  et~al\mbox{.}}{2003}]%
        {Kempe03}
\bibfield{author}{\bibinfo{person}{David Kempe}, \bibinfo{person}{Jon
  Kleinberg}, {and} \bibinfo{person}{{\'E}va Tardos}.}
  \bibinfo{year}{2003}\natexlab{}.
\newblock \showarticletitle{Maximizing the spread of influence through a social
  network}. In \bibinfo{booktitle}{\emph{SIGKDD}}. \bibinfo{pages}{137--146}.
\newblock


\bibitem[\protect\citeauthoryear{Khalil, Dilkina, and Song}{Khalil
  et~al\mbox{.}}{2014}]%
        {khalil2014scalable}
\bibfield{author}{\bibinfo{person}{Elias~Boutros Khalil},
  \bibinfo{person}{Bistra Dilkina}, {and} \bibinfo{person}{Le Song}.}
  \bibinfo{year}{2014}\natexlab{}.
\newblock \showarticletitle{Scalable diffusion-aware optimization of network
  topology}. In \bibinfo{booktitle}{\emph{SIGKDD}}.
  \bibinfo{pages}{1226--1235}.
\newblock


\bibitem[\protect\citeauthoryear{Leskovec, Krause, Guestrin, Faloutsos,
  VanBriesen, and Glance}{Leskovec et~al\mbox{.}}{2007}]%
        {leskovec2007cost}
\bibfield{author}{\bibinfo{person}{Jure Leskovec}, \bibinfo{person}{Andreas
  Krause}, \bibinfo{person}{Carlos Guestrin}, \bibinfo{person}{Christos
  Faloutsos}, \bibinfo{person}{Jeanne VanBriesen}, {and}
  \bibinfo{person}{Natalie Glance}.} \bibinfo{year}{2007}\natexlab{}.
\newblock \showarticletitle{Cost-effective outbreak detection in networks}. In
  \bibinfo{booktitle}{\emph{SIGKDD}}. \bibinfo{pages}{420--429}.
\newblock


\bibitem[\protect\citeauthoryear{Li, Smith, Dinh, and Thai}{Li
  et~al\mbox{.}}{2019}]%
        {li2019tiptop}
\bibfield{author}{\bibinfo{person}{Xiang Li}, \bibinfo{person}{J~David Smith},
  \bibinfo{person}{Thang~N Dinh}, {and} \bibinfo{person}{My~T Thai}.}
  \bibinfo{year}{2019}\natexlab{}.
\newblock \showarticletitle{Tiptop:(almost) exact solutions for influence
  maximization in billion-scale networks}.
\newblock \bibinfo{journal}{\emph{IEEE/ACM Transactions on Networking}}
  \bibinfo{volume}{27}, \bibinfo{number}{2} (\bibinfo{year}{2019}),
  \bibinfo{pages}{649--661}.
\newblock


\bibitem[\protect\citeauthoryear{Liu, Chen, Jeon, Chen, and Chen}{Liu
  et~al\mbox{.}}{2019}]%
        {liu2019influence}
\bibfield{author}{\bibinfo{person}{Wei Liu}, \bibinfo{person}{Xin Chen},
  \bibinfo{person}{Byeungwoo Jeon}, \bibinfo{person}{Ling Chen}, {and}
  \bibinfo{person}{Bolun Chen}.} \bibinfo{year}{2019}\natexlab{}.
\newblock \showarticletitle{Influence maximization on signed networks under
  independent cascade model}.
\newblock \bibinfo{journal}{\emph{Applied Intelligence}} \bibinfo{volume}{49},
  \bibinfo{number}{3} (\bibinfo{year}{2019}), \bibinfo{pages}{912--928}.
\newblock


\bibitem[\protect\citeauthoryear{L{\"u}, Zhou, Zhang, and Stanley}{L{\"u}
  et~al\mbox{.}}{2016}]%
        {lu2016h}
\bibfield{author}{\bibinfo{person}{Linyuan L{\"u}}, \bibinfo{person}{Tao Zhou},
  \bibinfo{person}{Qian-Ming Zhang}, {and} \bibinfo{person}{H~Eugene Stanley}.}
  \bibinfo{year}{2016}\natexlab{}.
\newblock \showarticletitle{The H-index of a network node and its relation to
  degree and coreness}.
\newblock \bibinfo{journal}{\emph{Nature communications}} \bibinfo{volume}{7},
  \bibinfo{number}{1} (\bibinfo{year}{2016}), \bibinfo{pages}{1--7}.
\newblock


\bibitem[\protect\citeauthoryear{Minutoli, Halappanavar, Kalyanaraman,
  Sathanur, Mcclure, and McDermott}{Minutoli et~al\mbox{.}}{2019}]%
        {minutoli2019fast}
\bibfield{author}{\bibinfo{person}{Marco Minutoli}, \bibinfo{person}{Mahantesh
  Halappanavar}, \bibinfo{person}{Ananth Kalyanaraman}, \bibinfo{person}{Arun
  Sathanur}, \bibinfo{person}{Ryan Mcclure}, {and} \bibinfo{person}{Jason
  McDermott}.} \bibinfo{year}{2019}\natexlab{}.
\newblock \showarticletitle{Fast and scalable implementations of influence
  maximization algorithms}. In \bibinfo{booktitle}{\emph{2019 IEEE
  International Conference on Cluster Computing (CLUSTER)}}.
  \bibinfo{pages}{1--12}.
\newblock


\bibitem[\protect\citeauthoryear{Nemhauser, Wolsey, and Fisher}{Nemhauser
  et~al\mbox{.}}{1978}]%
        {nemhauser1978analysis}
\bibfield{author}{\bibinfo{person}{George~L Nemhauser},
  \bibinfo{person}{Laurence~A Wolsey}, {and} \bibinfo{person}{Marshall~L
  Fisher}.} \bibinfo{year}{1978}\natexlab{}.
\newblock \showarticletitle{An analysis of approximations for maximizing
  submodular set functions}.
\newblock \bibinfo{journal}{\emph{Mathematical programming}}
  \bibinfo{volume}{14}, \bibinfo{number}{1} (\bibinfo{year}{1978}),
  \bibinfo{pages}{265--294}.
\newblock


\bibitem[\protect\citeauthoryear{Newman}{Newman}{2002}]%
        {newman2002spread}
\bibfield{author}{\bibinfo{person}{Mark~EJ Newman}.}
  \bibinfo{year}{2002}\natexlab{}.
\newblock \showarticletitle{Spread of epidemic disease on networks}.
\newblock \bibinfo{journal}{\emph{Physical review E}} \bibinfo{volume}{66},
  \bibinfo{number}{1} (\bibinfo{year}{2002}), \bibinfo{pages}{016128}.
\newblock


\bibitem[\protect\citeauthoryear{Ohsaka, Akiba, Yoshida, and
  Kawarabayashi}{Ohsaka et~al\mbox{.}}{2014}]%
        {ohsaka2014fast}
\bibfield{author}{\bibinfo{person}{Naoto Ohsaka}, \bibinfo{person}{Takuya
  Akiba}, \bibinfo{person}{Yuichi Yoshida}, {and} \bibinfo{person}{Ken-ichi
  Kawarabayashi}.} \bibinfo{year}{2014}\natexlab{}.
\newblock \showarticletitle{Fast and Accurate Influence Maximization on Large
  Networks with Pruned Monte-Carlo Simulations}. In
  \bibinfo{booktitle}{\emph{AAAI}}. \bibinfo{pages}{138--144}.
\newblock


\bibitem[\protect\citeauthoryear{Shu, Wang, Tang, and Do}{Shu
  et~al\mbox{.}}{2015}]%
        {shu2015numerical}
\bibfield{author}{\bibinfo{person}{Panpan Shu}, \bibinfo{person}{Wei Wang},
  \bibinfo{person}{Ming Tang}, {and} \bibinfo{person}{Younghae Do}.}
  \bibinfo{year}{2015}\natexlab{}.
\newblock \showarticletitle{Numerical identification of epidemic thresholds for
  susceptible-infected-recovered model on finite-size networks}.
\newblock \bibinfo{journal}{\emph{Chaos: An Interdisciplinary Journal of
  Nonlinear Science}} \bibinfo{volume}{25}, \bibinfo{number}{6}
  (\bibinfo{year}{2015}), \bibinfo{pages}{063104}.
\newblock


\bibitem[\protect\citeauthoryear{Sun, Huang, Yu, and Chen}{Sun
  et~al\mbox{.}}{2018}]%
        {sun2018multi}
\bibfield{author}{\bibinfo{person}{Lichao Sun}, \bibinfo{person}{Weiran Huang},
  \bibinfo{person}{Philip~S Yu}, {and} \bibinfo{person}{Wei Chen}.}
  \bibinfo{year}{2018}\natexlab{}.
\newblock \showarticletitle{Multi-round influence maximization}. In
  \bibinfo{booktitle}{\emph{SIGKDD}}. \bibinfo{pages}{2249--2258}.
\newblock


\bibitem[\protect\citeauthoryear{Tang, Shi, and Xiao}{Tang
  et~al\mbox{.}}{2015}]%
        {tang2015influence}
\bibfield{author}{\bibinfo{person}{Youze Tang}, \bibinfo{person}{Yanchen Shi},
  {and} \bibinfo{person}{Xiaokui Xiao}.} \bibinfo{year}{2015}\natexlab{}.
\newblock \showarticletitle{Influence maximization in near-linear time: A
  martingale approach}. In \bibinfo{booktitle}{\emph{SIGMOD}}.
  \bibinfo{pages}{1539--1554}.
\newblock


\bibitem[\protect\citeauthoryear{Tang, Xiao, and Shi}{Tang
  et~al\mbox{.}}{2014}]%
        {tang2014influence}
\bibfield{author}{\bibinfo{person}{Youze Tang}, \bibinfo{person}{Xiaokui Xiao},
  {and} \bibinfo{person}{Yanchen Shi}.} \bibinfo{year}{2014}\natexlab{}.
\newblock \showarticletitle{Influence maximization: Near-optimal time
  complexity meets practical efficiency}. In
  \bibinfo{booktitle}{\emph{SIGMOD}}. \bibinfo{pages}{75--86}.
\newblock


\bibitem[\protect\citeauthoryear{Wang, Fan, Li, and Tan}{Wang
  et~al\mbox{.}}{2017}]%
        {wang2017real}
\bibfield{author}{\bibinfo{person}{Yanhao Wang}, \bibinfo{person}{Qi Fan},
  \bibinfo{person}{Yuchen Li}, {and} \bibinfo{person}{Kian-Lee Tan}.}
  \bibinfo{year}{2017}\natexlab{}.
\newblock \showarticletitle{Real-time influence maximization on dynamic social
  streams}.
\newblock \bibinfo{journal}{\emph{PVLDB}} \bibinfo{volume}{10},
  \bibinfo{number}{7} (\bibinfo{year}{2017}), \bibinfo{pages}{805--816}.
\newblock


\bibitem[\protect\citeauthoryear{Wu and K{\"u}{\c{c}}{\"u}kyavuz}{Wu and
  K{\"u}{\c{c}}{\"u}kyavuz}{2018}]%
        {wu2018two}
\bibfield{author}{\bibinfo{person}{Hao-Hsiang Wu} {and} \bibinfo{person}{Simge
  K{\"u}{\c{c}}{\"u}kyavuz}.} \bibinfo{year}{2018}\natexlab{}.
\newblock \showarticletitle{A two-stage stochastic programming approach for
  influence maximization in social networks}.
\newblock \bibinfo{journal}{\emph{Computational Optimization and Applications}}
  \bibinfo{volume}{69}, \bibinfo{number}{3} (\bibinfo{year}{2018}),
  \bibinfo{pages}{563--595}.
\newblock


\bibitem[\protect\citeauthoryear{Xie, Meng, Sun, Ma, Yan, and Hu}{Xie
  et~al\mbox{.}}{2021}]%
        {xie2021detecting}
\bibfield{author}{\bibinfo{person}{Jiarong Xie}, \bibinfo{person}{Fanhui Meng},
  \bibinfo{person}{Jiachen Sun}, \bibinfo{person}{Xiao Ma},
  \bibinfo{person}{Gang Yan}, {and} \bibinfo{person}{Yanqing Hu}.}
  \bibinfo{year}{2021}\natexlab{}.
\newblock \showarticletitle{Detecting and modelling real percolation and phase
  transitions of information on social media}.
\newblock \bibinfo{journal}{\emph{Nature Human Behaviour}}
  (\bibinfo{year}{2021}), \bibinfo{pages}{1--8}.
\newblock


\bibitem[\protect\citeauthoryear{Yang, Chen, Gao, and Yan}{Yang
  et~al\mbox{.}}{2020}]%
        {yang2020boosting}
\bibfield{author}{\bibinfo{person}{Wenguo Yang}, \bibinfo{person}{Shengminjie
  Chen}, \bibinfo{person}{Suixiang Gao}, {and} \bibinfo{person}{Ruidong Yan}.}
  \bibinfo{year}{2020}\natexlab{}.
\newblock \showarticletitle{Boosting node activity by recommendations in social
  networks}.
\newblock \bibinfo{journal}{\emph{Journal of Combinatorial Optimization}}
  \bibinfo{volume}{40} (\bibinfo{year}{2020}), \bibinfo{pages}{825--847}.
\newblock


\bibitem[\protect\citeauthoryear{Yang, Ma, Li, Yan, Yuan, Wu, and Li}{Yang
  et~al\mbox{.}}{2019}]%
        {yang2019marginal}
\bibfield{author}{\bibinfo{person}{Wenguo Yang}, \bibinfo{person}{Jianmin Ma},
  \bibinfo{person}{Yi Li}, \bibinfo{person}{Ruidong Yan}, \bibinfo{person}{Jing
  Yuan}, \bibinfo{person}{Weili Wu}, {and} \bibinfo{person}{Deying Li}.}
  \bibinfo{year}{2019}\natexlab{}.
\newblock \showarticletitle{Marginal gains to maximize content spread in social
  networks}.
\newblock \bibinfo{journal}{\emph{IEEE Transactions on Computational Social
  Systems}} \bibinfo{volume}{6}, \bibinfo{number}{3} (\bibinfo{year}{2019}),
  \bibinfo{pages}{479--490}.
\newblock


\bibitem[\protect\citeauthoryear{Zhang, Zhou, Tao, Karras, Li, and Xiong}{Zhang
  et~al\mbox{.}}{2020}]%
        {zhang2020geodemographic}
\bibfield{author}{\bibinfo{person}{Kaichen Zhang}, \bibinfo{person}{Jingbo
  Zhou}, \bibinfo{person}{Donglai Tao}, \bibinfo{person}{Panagiotis Karras},
  \bibinfo{person}{Qing Li}, {and} \bibinfo{person}{Hui Xiong}.}
  \bibinfo{year}{2020}\natexlab{}.
\newblock \showarticletitle{Geodemographic influence maximization}. In
  \bibinfo{booktitle}{\emph{SIGKDD}}. \bibinfo{pages}{2764--2774}.
\newblock


\end{thebibliography}
